\documentclass[runningheads]{llncs}
\usepackage{graphicx} 
\usepackage[T1]{fontenc}
\usepackage{stmaryrd} 
\usepackage{xspace}
\usepackage{cite}
\usepackage{url}
\usepackage{listings}
\usepackage{subcaption}
\usepackage{graphicx}
\usepackage{float}
\usepackage{wrapfig}
\usepackage{needspace}
\usepackage[table, dvipsnames]{xcolor}
\usepackage{color}
\usepackage{courier}
\usepackage{multirow}
\usepackage{amssymb}
\usepackage{hhline}
\usepackage{hyperref}
\usepackage{titletoc}
\usepackage{mathpartir}

\providecommand{\doi}[1]{\url{https://doi.org/#1}}

\newcommand{\sys}{\textbf{\textbf{crush}}\xspace}
\newcommand{\leanauto}{\leaninline{auto}\xspace}
\newcommand{\grind}{\leaninline{grind}\xspace}
\newcommand{\leansmt}{lean-smt\xspace}
\newcommand{\EG}{e.g.,\xspace}
\newcommand{\IE}{i.e.,\xspace}

\definecolor{LeanKeyword}{RGB}{0,82,155}
\definecolor{LeanDeclaration}{RGB}{111,45,168}
\definecolor{LeanType}{RGB}{0,112,105}
\definecolor{LeanTactic}{RGB}{166,86,0}
\definecolor{LeanComment}{RGB}{63,122,52}
\definecolor{LeanString}{RGB}{177,43,74}

\NewDocumentCommand{\mike} { mO{} }{\textcolor{blue}{\textsuperscript{\textit{Mike}}\textsf{\textbf{\small[#1]}}}}
\NewDocumentCommand{\dexin} { mO{} }{\textcolor{gray}{\textsuperscript{\textit{Dexin}}\textsf{\textbf{\small[#1]}}}}
\NewDocumentCommand{\aarti} { mO{} }{\textcolor{purple}{\textsuperscript{\textit{Aarti}}\textsf{\textbf{\small[#1]}}}}

\lstdefinelanguage{Lean}{
  sensitive=true,
  alsoletter={_},
  morekeywords=[1]{
    abbrev,axiom,class,def,deriving,end,example,export,import,inductive,
    crush_lower,crush_lower_result,crush_translate,crush_translate_sort,
    crush_unfold,crush_reconstruct,
    instance,namespace,opaque,open,register_crush_replay,register_lowering,section,set_option,
    structure,theorem,universe,variable,variables,where
  },
  keywordstyle=[1]{\color{LeanDeclaration}\bfseries},
  morekeywords=[2]{
    as,by,do,else,for,from,fun,have,if,in,let,match,mutual,nomatch,
    return,show,then,try,unless,while,with
  },
  keywordstyle=[2]{\color{LeanKeyword}\bfseries},
  morekeywords=[3]{
    Array,Bool,Char,False,Fin,Float,Int,IO,List,LoweringHandler,Nat,None,
    Option,Prod,Prop,Result,some,Sort,SortHandler,String,Sum,TranslationHandler,
    True,Type,Unit,false,true
  },
  keywordstyle=[3]{\color{LeanType}\bfseries},
  morekeywords=[4]{
    aesop,apply,assumption,auto,cases,constructor,contradiction,crush,
    decide,exact,grind,induction,intro,intros,norm_num,omega,rfl,rw,
    simp,simp_all,simpa,solve_by_elim,specialize,subst
  },
  keywordstyle=[4]{\color{LeanTactic}\bfseries},
  morecomment=[l]{--},
  morecomment=[s]{/-}{-/},
  commentstyle={\color{LeanComment}},
  morestring=[b]",
  stringstyle={\color{LeanString}}
}

\lstdefinelanguage{Alethe}{
  sensitive=true,
  alsoletter={_,-,:},
  morekeywords=[1]{assume,step,anchor,cl},
  keywordstyle=[1]{\color{LeanDeclaration}\bfseries},
  morekeywords=[2]{:rule,:premises,:args,:discharge,:named},
  keywordstyle=[2]{\color{LeanKeyword}\bfseries},
  morekeywords=[3]{
    Bool,Int,Real,BitVec,Array,true,false,not,and,or,xor,ite,let,forall,exists,
    bvnot,bvand,bvor,bvxor,bvnand,bvnor,bvxnor,bvneg,bvadd,bvsub,bvmul,
    bvult,bvule,bvslt,bvsle,concat,extract,select,store
  },
  keywordstyle=[3]{\color{LeanType}\bfseries},
  morekeywords=[4]{
    refl,cong,trans,symm,resolution,rare_rewrite,reordering,contraction,
    weakening,forall_inst,bind,subproof,evaluate,aci_simp
  },
  keywordstyle=[4]{\color{LeanTactic}\bfseries},
  morecomment=[l]{;},
  commentstyle={\color{LeanComment}},
  morestring=[d]",
  stringstyle={\color{LeanString}}
}

\lstdefinestyle{lean}{
  language=Lean,
  basicstyle=\ttfamily\scriptsize,
  columns=fullflexible,
  breakatwhitespace=false,
  breaklines=true,
  numbers=none,
  captionpos=b,
  keepspaces=true,
  xleftmargin=0.2em,
  showspaces=false,
  showstringspaces=false,
  showtabs=false,
  tabsize=2,
  numbersep=5pt,
  literate=
    {→}{{$\to$}}1
    {←}{{$\leftarrow$}}1
    {↔}{{$\leftrightarrow$}}1
    {⇒}{{$\Rightarrow$}}1
    {∀}{{$\forall$}}1
    {∃}{{$\exists$}}1
    {λ}{{$\lambda$}}1
    {≠}{{$\ne$}}1
    {≤}{{$\le$}}1
    {≥}{{$\ge$}}1
    {∧}{{$\land$}}1
    {∨}{{$\lor$}}1
    {¬}{{$\neg$}}1
    {⊢}{{$\vdash$}}1
}

\lstdefinestyle{inlinecode}{
  basicstyle=\ttfamily,
  columns=fullflexible,
  breaklines=false,
  breakatwhitespace=false,
  breakautoindent=false,
  breakindent=0pt,
  postbreak={},
  mathescape=false,
  escapeinside={},
  frame=none,
  numbers=none
}
\DeclareRobustCommand{\leaninline}[1]{%
  \mbox{\lstinline[style=lean,style=inlinecode]{#1}}}
\DeclareRobustCommand{\aletheinline}[1]{%
  \mbox{\lstinline[style=lean,style=inlinecode,language=Alethe]{#1}}}

\newcommand{\myparagraph}[1]{\vspace*{3pt}\noindent\textbf{\emph{#1}}\xspace}

\spnewtheorem*{proofsketch}{Proof sketch}{\itshape}{\rmfamily}

\newcommand{\sembrackets}[1]{\llbracket #1\rrbracket}
\newcommand{\denote}[3]{\sembrackets{#1}_{#2}(#3)}
\newcommand{\typedenote}[2]{\sembrackets{#1}_{#2}}
\newcommand{\booltype}{\mathsf{Bool}}
\newcommand{\proptype}{\mathsf{Prop}}
\newcommand{\defunc}{\mathcal F}
\newcommand{\theoryencoding}{\mathcal D}
\newcommand{\headapp}{\mathcal A}
\newcommand{\bodytrans}{\mathcal B}
\newcommand{\fnsort}[1]{%
  \mathsf{Fn}\if\relax\detokenize{#1}\relax\else_{#1}\fi}
\newcommand{\appsym}[1]{\mathsf{app}_{#1}}
\newcommand{\clossym}[1]{\mathsf{clos}_{#1}}
\newcommand{\extaxiom}[1]{\mathsf{ext}_{#1}}
\newcommand{\figurepart}[1]{\par\noindent{\raggedright\bfseries #1\par}}

\newcommand{\DoToC}{%
  \startcontents[appendices]
  \begingroup
  \def\authcount##1{}%
  \printcontents[appendices]{}{1}{%
    \setcounter{tocdepth}{2}%
    \textbf{Table of Contents}\par
    \vskip3pt\hrule\vskip1pt\hrule\vskip5pt
  }%
  \par\vskip3pt\hrule\vskip1pt\hrule\vskip5pt
  \endgroup
}

\begin{document}
\title{\sys: Extensible Proof Automation for Lean}
\author{Mike He\and Dexin Zhang\and Aarti Gupta}
\institute{Princeton University}

\maketitle

\begin{abstract}
Lean's rich expressiveness supports higher-order functions, inductive datatypes,
and user-defined abstractions, making it challenging to translate proof
obligations into first-order Satisfiability Modulo Theory (SMT) queries and reconstruct kernel-checked proofs.
We present \sys, an extensible SMT-backed proof automation framework for Lean.
To bridge the gap between Lean's higher-order constructs and the solver's
first-order logic, \sys uses defunctionalization to encode function values
and their applications. We formalize this translation for a supported
higher-order fragment and prove its soundness.
Moreover, \sys provides metaprogramming interfaces for extending both
Lean-to-SMT translation and the reconstruction of Lean proofs from solver certificates without modifying the translation pipeline.
We implement \sys as both a tactic and a backend for other third-party
automation systems.
We show that across four such use cases that target
verification and general proof
automation, \sys discharges more obligations than the competing backends and reconstructs more kernel-checked proofs than the baselines.

\keywords{Interactive Theorem Proving\and SMT Solvers\and Lean Hammer}
\end{abstract}

\section{Introduction}
Theorem proving has been widely applied in verification of hardware designs~\cite{choi2017kami, russinoff1998mechanically, kaivola2009intelfv}, software codebases~\cite{appel2011vst, leroy2016compcert, seL4} and cloud services~\cite{ironfleet, cacm-aws-formal-semi-formal, grove2023sosp}. 
The success of theorem proving in these domains can be attributed to rigorous improvements along two lines: interactive theorem proving and automated theorem provers (ATP).
Interactive theorem proving involves manual formalization and proofs, checked by a proof assistant (PA) such as Rocq~\cite{coq-prover}, Lean~\cite{lean-prover}, Isabelle/HOL~\cite{isabelle-hol} and Agda~\cite{agda-prover}.
These tools provide rich type systems in which users can formalize properties with abstract and higher-order constructs.
On the other hand, automated theorem provers such as Z3~\cite{z3}, CVC5~\cite{cvc5}, and E-Prover~\cite{e-prover} support less expressive logics but offer significant automation in proof search. 

\newcommand{\supported}{\ensuremath{\checkmark}}
\newcommand{\partiallysupported}{\ensuremath{\triangle}}
\newcommand{\unsupported}{\ensuremath{\times}}

\begin{table}[t!]
\centering
\caption{Qualitative comparison of goal translation and proof reconstruction.
\emph{Poly.} shows support of polymorphic premises; \emph{HO} shows support of higher-order functions;
\emph{Ext. encoding} denotes interfaces for custom encodings.
\emph{Core-guided} and \emph{Replay} shows whether the tool supports reconstructing ITP-checked proofs from solver certificates (unsat cores and solver-produced proofs, respectively).
}
\label{tab:translation-comparison}
\scriptsize
\setlength{\tabcolsep}{4pt}
\renewcommand{\arraystretch}{1.15}
\begin{tabular}{@{}llccc|cc@{}}
\hline
& & \multicolumn{3}{c|}{\textbf{ITP $\rightarrow$ ATP}}
& \multicolumn{2}{c}{\textbf{ATP $\rightarrow$ ITP}} \\
\cline{3-7}
\textbf{ITP} & \textbf{Tool}
& \textbf{Poly.}
& \textbf{HO}
& \textbf{Ext. encoding}
& \textbf{Core-guided}
& \textbf{Replay} \\
\hline

Isabelle/HOL
& \texttt{smt}~\cite{bohme2010fast}
& \supported
& \supported
& \supported
& \unsupported
& \supported \\
\hline

\multirow{2}{*}{Rocq}
& SMTCoq~\cite{smtcoq}
& \unsupported
& \unsupported
& \supported
& \unsupported
& \supported \\

& CoqHammer~\cite{coqhammer}
& \supported
& \supported
& \unsupported
& \supported
& \unsupported \\
\hline

\multirow{4}{*}{Lean}
& lean-auto~\cite{lean-auto}\textsuperscript{*}
& \supported
& \unsupported
& \unsupported
& \unsupported
& \unsupported \\

& LeanHammer~\cite{leanhammer}\textsuperscript{**}
& \supported
& \supported
& \unsupported
& \supported
& \unsupported \\

& lean-smt~\cite{lean-smt}
& \supported
& \unsupported
& \supported
& \unsupported
& \supported \\
\cline{2-7}

& \sys
& \supported
& \supported
& \supported
& \supported
& \supported \\
\hline
\end{tabular}
\par\smallskip
\begin{minipage}{\linewidth}
\scriptsize
\textsuperscript{*}SMT backend of lean-auto.\par
\textsuperscript{**}Duper backend of LeanHammer; its replay support is provided separately by lean-smt.
\end{minipage}
\end{table}

Integrating ATPs with proof assistants, as in CoqHammer~\cite{coqhammer} and
SMTCoq~\cite{smtcoq}, combines automated search with kernel-checked proofs.
While offering merits from both PA and ATP, these integrations face two main challenges, shown in
Table~\ref{tab:translation-comparison}: (1) translating expressive proof
obligations into the solver's first-order logic and (2) reconstructing kernel-checkable
proofs from solvers' outputs.
The translation has to accommodate features like polymorphism and higher-order constructs admitted by PAs 
and provide extensible encodings that allow users to customize how new definitions are translated.
Proof reconstruction has to produce kernel-checkable derivations from certificates generated by the solvers,
which may provide only selected premises (\IE unsat cores), or rely on solver inference rules that have
to be mapped into the PA's proof system.

In Lean specifically, proof automation that leverages power of ATPs follows two broad approaches.
The first consists of tactics implemented natively in Lean, represented by \grind~\cite{grind-tactic} and \leaninline{aesop}~\cite{aesop-tactic}.
These tactics operate directly on Lean expressions and construct proof terms that are checked by the Lean kernel.
\grind combines congruence closure, E-matching, constraint propagation, case analysis, and specialized solvers for domains such as linear arithmetic and commutative rings.
\leaninline{aesop}, in contrast, constructs a best-first proof-search tree from an extensible set of rules.
Rules may be built from theorem applications, constructor introduction, term rewriting, or arbitrary tactics.
These tactics are native to Lean, but their automation is
limited to the search procedures, registered rules, and decision procedures
implemented within Lean.

The second category consists of hammer-style tactics, such as lean-auto~\cite{lean-auto} and \leansmt~\cite{lean-smt}.
These tactics collect local hypotheses and the negated proof goal, translate the supported fragment of Lean into a solver language such as SMT-LIB~\cite{smtlib2}, and invoke an external SMT solver.
This architecture allows Lean to benefit from mature decision procedures and proof-search heuristics developed for ATPs.
When the SMT solver reports unsatisfiability (UNSAT), this fact could either be trusted directly, or be witnessed by a reconstructed proof, using solver's certificates, that is subsequently and independently checked by the Lean kernel.
The translation boundary of existing approaches suffers from 
a limitation: Lean constructs, including higher-order functions and
user-defined abstractions, must be represented in the solver's supported logic.
The applicability and soundness guarantees
of these tactics are therefore limited by
the coverage and extensibility of both the translation and proof reconstruction procedures.



To address these gaps, we have developed \sys, an SMT-backed proof automation framework
that combines translation support for higher-order functions 
with user-extensible encodings, and
supports both unsat core-guided proof reconstruction and proof replay from solver inference steps.
Users can customize how their definitions are encoded and how the resulting
solver proofs are reconstructed in Lean.
Our design has three key capabilities.
First, \sys supports nondependent higher-order functions, including
function-valued arguments, partial applications, and lambda abstractions. Through defunctionalization, it represents these
functions as values of uninterpreted SMT sorts and introduces application
symbols and equations describing their behavior.
We formalize this translation and prove its soundness.
Second, we provide an extensible translation interface that allows users to define customized encodings of Lean types, functions, and operators without modifying the core translator.
Third, we provide an extensible proof replay interface that allows users to define proof reconstruction procedures in the form of typed rewrite rules for terms introduced by their custom translation and for solver's inference
rules, yielding kernel-checkable proofs.

We evaluate \sys on four benchmark suites: a Curated set of proof obligations developed by us,
Cashmere~\cite{loom, Cashmere-github},
Velvet~\cite{velvet}, and PLean~\cite{plean}. These include 754 verification conditions (VCs) that contain mixed theories.
\sys discharges 709 (94.0\%) of the 754 VCs, with 693 (91.9\%) successfully reconstructed kernel-checked proofs.
In comparison, \leanauto, the strongest non-proof-producing tactic, closes only 622 (82.5\%), and \grind, the strongest proof-producing tactic, closes 641 (85.0\%).
With kernel-checked proofs, \sys also exceeds \leansmt~\cite{lean-smt} and Duper~\cite{duper} by 34.7 and 37.1 percentage points in overall coverage, respectively.

\noindent In summary, our contributions are:

\myparagraph{Extensible translation.}
We design a pattern-based interface for extending SMT encodings of Lean types
and terms.
These extensions combine with translation of higher-order constructs with defunctionalization
without modifying the core translator.
We also formalized soundness of the defunctionalization pass.

\myparagraph{Extensible proof replay.}
We design interfaces for interpreting terms in SMT solver proofs and reconstructing their inference steps in Lean. Users can extend the procedure to support custom encodings and additional solver inference rules.


\myparagraph{The \sys framework.}
We implement these ideas as a proof automation framework \sys in Lean.
\sys is publicly available on GitHub\footnote{\href{https://github.com/AD1024/lean-crush}{https://github.com/AD1024/lean-crush}}.

\section{Overview}
\label{sec:overview}

Fig.~\ref{fig:crush-pipeline} summarizes the end-to-end \sys pipeline.
Given a Lean goal, its hypotheses, and lemmas, 
\sys first prepares a set of propositions by negating the goal,
collecting all local hypotheses, and adding user-supplied
lemmas.
The pre-process pass performs normalization, monomorphization and a bounded quantifier instantiation on the collected formulas.
Appendix~\ref{app:query-preparation} describes these standard pre-processing steps in more detail.
\sys then translates processed formulas into an SMT query, where hypotheses and lemmas are conjoined with the negated goal, and employs defunctionalization to encode higher-order functions and their applications in first-order logic.
If the solver reports UNSAT, \sys can either
trust the solver or construct a Lean proof from an unsat core or by replaying the inference steps produced by the solver.
If solver reports SAT (with a model) or Unknown, 
the user may leverage external tools to fix or refine the proof.
\begin{figure}[t!]
  \centering
  \includegraphics[width=0.90\textwidth]{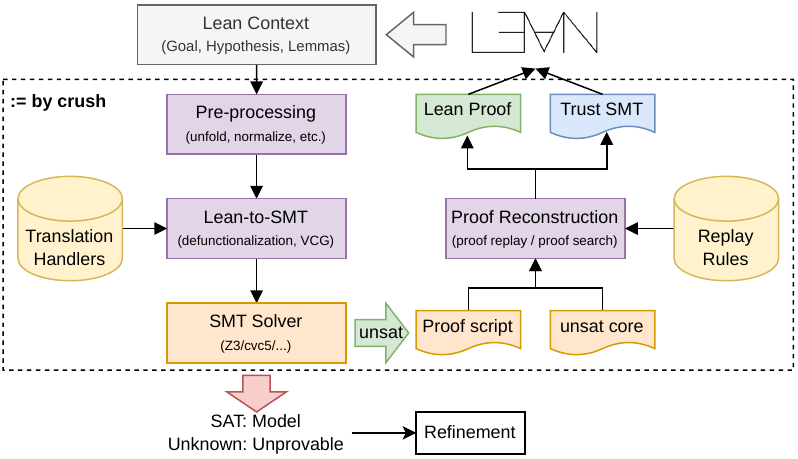}
  \caption{The \sys pipeline.}
  \label{fig:crush-pipeline}
\end{figure}


\Needspace{15\baselineskip}
\subsection{Running Example}

\begin{wrapfigure}[12]{r}{0.54\textwidth}
  \vspace{-\intextsep}
  \centering
  \captionsetup{skip=4pt}
  \begin{lstlisting}[
    frame=single,
    aboveskip=3pt,
    belowskip=3pt,
    numbers=left,
    name=overview,
    firstnumber=1,
    numberstyle=\tiny\color{gray},
    numbersep=6pt,
    xleftmargin=2em,
    framexleftmargin=1.5em,
    escapeinside={(*@}{@*)}]
def Nand (x y : BitVec 4) : BitVec 4 :=(*@\label{line:overview-operator-start}@*)
  ~~~(x &&& y)(*@\label{line:overview-operator-end}@*)

theorem computes_masked(*@\label{line:overview-theorem-start}@*)
    (run : (BitVec 4 → BitVec 4) →(*@\label{line:overview-callback}@*)
      BitVec 4 → BitVec 4)
    (hrun : ∀ f x, run f x = f x)(*@\label{line:overview-spec}@*)
    (x mask : BitVec 4) :
    run (fun y => Nand (Nand y mask)(*@\label{line:overview-callback-body}@*)
                      (Nand y mask)) x =
      x &&& mask := by(*@\label{line:overview-goal}@*)
  crush(*@\label{line:overview-tactic}@*)
\end{lstlisting}
  \caption{The running example: a callback that applies a bitvector mask.}
  \label{fig:overview-example}
\end{wrapfigure}

Fig.~\ref{fig:overview-example} defines a function NAND over two bitvectors $x, y$ (with width $4$),  
and states a property for \sys to prove. The hypothesis on
line~\ref{line:overview-spec} specifies that applying \leaninline{run} on function \leaninline{f} with \leaninline{x} gives the result of applying \leaninline{f} on \leaninline{x}.
The function on line~\ref{line:overview-callback-body} applies a mask to a given input $y$. Its two inner NAND applications produce the same complemented conjunction; the outer NAND recovers the conjunction of the input and mask. Thus, the left-hand-side application computes the
masked input.

\Needspace{9\baselineskip}
\subsection{Extensible Translation}

\begin{wrapfigure}[7]{r}{0.40\textwidth}
  \vspace{-\intextsep}
  \centering
  \captionsetup{skip=4pt}
  \begin{lstlisting}[
    frame=single,
    aboveskip=3pt,
    belowskip=3pt,
    numbers=left,
    name=overview,
    firstnumber=last,
    numberstyle=\tiny\color{gray},
    numbersep=6pt,
    xleftmargin=2em,
    framexleftmargin=1.5em,
    escapeinside={(*@}{@*)}]
register_lowering term <<(*@\label{line:overview-lowering-start}@*)
  (Nand (term x) (term y)) =>(*@\label{line:overview-lowering-pattern}@*)
    (smt| (bvnand $x $y))(*@\label{line:overview-lowering-quotation}@*)
>>
\end{lstlisting}
  \caption{A custom encoding from \texttt{Nand} to solver's \texttt{bvnand}.}
  \label{fig:overview-lowering}
\end{wrapfigure}

\sys supports extensible translation via rewrite rules to define SMT encodings for Lean terms.
Fig.~\ref{fig:overview-lowering} shows an example encoding for \leaninline{Nand}: 
the pattern on line~\ref{line:overview-lowering-pattern} matches applications of \leaninline{Nand}.
The sub-patterns \leaninline{(term x)} and \leaninline{(term y)} match the arguments and bind the translated results to \leaninline{x} and \leaninline{y}.
The right hand side pattern (line~\ref{line:overview-lowering-quotation}) defines the emitted SMT encoding that applies the solver's \leaninline{bvnand} operation with \leaninline{x} and \leaninline{y}, written in a DSL of an SMT-LIB subset embedded in Lean.
\S\ref{sec:translation} discusses this interface in more detail.

\subsection{Higher-Order Functions}

The translation is defined for bitvector operations, but the higher-order function
must also be represented in first-order SMT. 
To encode function values in first-order SMT, \sys introduces a
fresh first-order sort for each function type and a fresh operator for its
application. Then, it emits axioms that assign interpretations to the sort with respect to its definition in Lean.
This defunctionalization combines with the extensible translation to encode higher-order functions
and bitvector operations into theories supported by the solver.
\S\ref{sec:higher-order} presents the encoding in detail and proves its soundness.

\Needspace{13\baselineskip}
\subsection{Extensible Proof Reconstruction}

\begin{wrapfigure}[10]{r}{0.52\textwidth}
  \vspace{-\intextsep}
  \centering
  \captionsetup{skip=4pt}
  \begin{lstlisting}[
    frame=single,
    aboveskip=3pt,
    belowskip=3pt,
    numbers=left,
    name=overview,
    firstnumber=last,
    numberstyle=\tiny\color{gray},
    numbersep=6pt,
    xleftmargin=2em,
    framexleftmargin=1.5em,
    escapeinside={(*@}{@*)}]
register_crush_replay term <<(*@\label{line:overview-replay-start}@*)
  (bvnand (term x : BitVec _)
          (term y : BitVec _)) =>
    ~~~(x &&& y)(*@\label{line:overview-replay-lean-term}@*)
>>(*@\label{line:overview-replay-end}@*)

set_option crush.backend "cvc5"(*@\label{line:overview-options-start}@*)
set_option crush.trust "reconstruct"
set_option crush.reconstruct "alethe"(*@\label{line:overview-options-end}@*)
\end{lstlisting}
  \caption{Extending proof replay and selecting Alethe reconstruction.}
  \label{fig:overview-replay}
\end{wrapfigure}

The resulting query can now be handled by an SMT
solver. If it is unsatisfiable,
\sys supports two paths to obtain a kernel-checked proof.
Reconstruction from an
unsat core selects the relevant Lean hypotheses and uses Lean automation tactics (\EG \leaninline{grind}, \leaninline{omega})
to re-derive a proof.
Alternatively, \sys can replay inference steps produced by the solver. 
\sys currently supports the Alethe format~\cite{schurr2021alethe}.
As an inverse of Lean-to-SMT translation, \sys also supports extensible proof reconstruction via rewrite rules that map solver inference steps to a Lean term or proof.

Fig.~\ref{fig:overview-replay} shows a rewrite rule for the running example:
the pattern on lines 2--3 matches applications of \leaninline{bvnand} operation and its arguments.
The sub-patterns (\EG \leaninline{(term x : BitVec _)}) on the arguments match bitvectors of arbitrary known widths.
On the right hand side (line 4), the pattern defines the translated Lean term.
Finally, lines 7--9 configure \sys to enable proof reconstruction from Alethe scripts. 
\S\ref{sec:reconstruction} discusses the reconstruction procedures and the replay extension interface in detail.

\section{Extensible Lean-to-SMT Translation}
\label{sec:translation}

\sys provides extensible translation from Lean terms and types to SMT terms and sorts: unfolding a definition exposes the defining equation encoded in the SMT query.
Additionally, users may define custom encoding of terms and types as rewrite rules that match head symbols of a term or type application. 

\Needspace{14\baselineskip}
\subsection{Definition Unfolding}

\begin{wrapfigure}[10]{r}{0.50\textwidth}
  \vspace{-\intextsep}
  \centering
  \captionsetup{skip=4pt}
  \begin{lstlisting}[
    frame=single,
    aboveskip=3pt,
    belowskip=3pt,
    numbers=left,
    firstnumber=1,
    numberstyle=\tiny\color{gray},
    numbersep=6pt,
    xleftmargin=2em,
    framexleftmargin=1.5em,
    escapeinside={(*@}{@*)}]
@[crush_unfold](*@\label{line:unfold-clamp-attr}@*)
def clamp (x lo hi : Int) : Int :=
  if x < lo then lo
  else if hi < x then hi else x

theorem clamp_bounds (x lo hi : Int)
    (h : lo ≤ hi) :
    lo ≤ clamp x lo hi ∧
      clamp x lo hi ≤ hi := by
  crush(*@\label{line:unfold-clamp-crush}@*)
\end{lstlisting}
  \caption{Definition unfolding on functions.}
  \label{fig:unfold-clamp}
\end{wrapfigure}

The \leaninline{crush_unfold} attribute makes the defining equations of a function available whenever that function is called.
Fig.~\ref{fig:unfold-clamp} shows an example of an integer clamp: the annotation on
line~\ref{line:unfold-clamp-attr} directs \sys to include the definition of
\leaninline{clamp} in the emitted query. Then, the SMT solver proves the goal by the hypothesis on line 7 and the definition of \leaninline{clamp}.

\subsection{Head- and Type-Indexed Term Translation}

Let $e^\flat$ denote the SMT term produced from a Lean expression $e$, and $\tau^\flat$ denote the SMT sort used to represent a Lean type $\tau$.

\begin{figure}[t]
  \centering
  \begin{minipage}{0.94\textwidth}
  \begin{lstlisting}[
    frame=single,
    numbers=left,
    firstnumber=1,
    numberstyle=\tiny\color{gray},
    numbersep=6pt,
    xleftmargin=2em,
    framexleftmargin=1.5em,
    escapeinside={(*@}{@*)}]
register_lowering result-type with ctx <<
  Bool => do(*@\label{line:result-lowering-pattern}@*)
    match ctx.fn, ctx.args with
    | .const ``Bool.not _, #[p] =>(*@\label{line:result-lowering-not}@*)
        return some (smt| (not $(← ctx.emitTerm p)))
    | .const ``Bool.xor _, #[p, q] =>(*@\label{line:result-lowering-xor}@*)
        return some (smt| (xor $(← ctx.emitTerm p) $(← ctx.emitTerm q)))
    | _, _ => return none(*@\label{line:result-lowering-decline}@*)
>>
\end{lstlisting}
  \end{minipage}
  \caption{Grouping Boolean operations in a type-indexed lowering.}
  \label{fig:result-lowering}
\end{figure}

\myparagraph{Head-indexed translation.}
Translations on Lean terms can be defined as rewrite rules that match the head of an application.
For example, in Fig.~\ref{fig:overview-lowering}, the pattern on
line~\ref{line:overview-lowering-pattern} matches applications of \leaninline{Nand}
and translates its two operands, producing $(\texttt{Nand}\;x\;y)^\flat
  = \mathsf{bvnand}(x^\flat,y^\flat)$.

\myparagraph{Type-indexed translation.}
Rewrites can also be grouped based on the type of Lean terms.
For instance, in Fig.~\ref{fig:result-lowering}, the \leaninline{Bool} pattern on
line~\ref{line:result-lowering-pattern} matches both \leaninline{Bool.not} and
\leaninline{Bool.xor}. The handler then distinguishes these operations by their
application heads and arguments.
Lines~\ref{line:result-lowering-not} and~\ref{line:result-lowering-xor}
recursively translate the operands of negation and exclusive-or and construct
the corresponding SMT terms.
Line~\ref{line:result-lowering-decline} ignores
other Boolean terms, and the translation procedure will try other rewrite rules if the triggered rule returns \leaninline{none}.

\subsection{Coordinating Sort and Term Translation}

When users define custom types $\tau$ in Lean, \sys has to handle translations for both terms with type $\tau$ and the encoding of $\tau$.
With the extensible interface, \sys allows users to define rewrite rules over type applications to define encodings for user-defined Lean types. 

For example, suppose the user defines a Lean type $\mathsf{TotalMap}\;\kappa\;\nu$, whose values map keys of type $\kappa$ to values of type $\nu$.
Fig.~\ref{fig:map-lowering} shows the custom encoding of $\mathsf{TotalMap}\;\kappa\;\nu$ into the $\mathsf{Array}$ sort in SMT. The semantics of rewrite rules over Lean types is similar to rewrite rules over Lean terms.
The left-hand-side pattern on line~\ref{line:map-lowering-sort} matches the key and
value types, and the right-hand-side pattern on line~3 defines the corresponding SMT encoding.
Additionally, rewrite rules on lines~\ref{line:map-lowering-get} and \ref{line:map-lowering-set}
define encodings for applications of $\mathsf{get}$ and $\mathsf{set}$ functions associated with $\mathsf{TotalMap}$, translating $\mathsf{get}$ to \aletheinline{select} and $\mathsf{set}$ to \aletheinline{store}.
\Needspace{16\baselineskip}
\begin{wrapfigure}[14]{r}{0.55\textwidth}
  \vspace{-\intextsep}
  \centering
  \captionsetup{skip=4pt}
  \begin{lstlisting}[
    frame=single,
    aboveskip=3pt,
    belowskip=3pt,
    numbers=left,
    firstnumber=1,
    numberstyle=\tiny\color{gray},
    numbersep=6pt,
    xleftmargin=2em,
    framexleftmargin=1.5em,
    escapeinside={(*@}{@*)}]
register_lowering sort <<
  (TotalMap (sort key) (sort value)) =>(*@\label{line:map-lowering-sort}@*)
    (smt| (Array $key $value))
>>

register_lowering term <<
  (TotalMap.get _ _ (term m) (term k)) =>(*@\label{line:map-lowering-get}@*)
    (smt| (select $m $k))
>>

register_lowering term <<
  (TotalMap.set _ _ (term m)(*@\label{line:map-lowering-set}@*)
    (term k) (term v)) =>
    (smt| (store $m $k $v))
>>
\end{lstlisting}
  \caption{Encoding of $\mathsf{TotalMap}$.}
  \label{fig:map-lowering}
\end{wrapfigure}

\sys performs sort-checking on queries emitted by handlers before invoking the solver.
This light-weight checking catches incompatible sort and term translations,
such as erroneously applying an array operation (\EG \aletheinline{select}) to an integer.
However, \sys does not check that the SMT encoding is semantically correct: the user
is responsible for ensuring the correctness of extended translations.
For this example, since $\mathsf{TotalMap}$ is encoded as SMT $\mathsf{Array}$s, the user should ensure that $\mathsf{TotalMap.get}$ and $\mathsf{TotalMap.set}$ are 
consistent with the axioms of the array theory (\IE  $\mathsf{get}(\mathsf{set}(m,k,v),k)=v$). 


\section{Higher-Order Translation via Defunctionalization}
\label{sec:higher-order}
\label{sec:metatheory}

\sys uses defunctionalization to encode higher-order functions in first-order
SMT, introducing uninterpreted sorts for function types and replacing function values and their applications with uninterpreted symbols constrained by axioms. 
Our translation procedure supports the encoding for a subset of Lean with \textit{monomorphic}, \textit{nondependent} higher-order constructs.
We model this translation procedure in Lean and formalize its soundness.

\subsection{Defunctionalization}
\label{sec:defunctionalization}

\myparagraph{Function values and application.}
For a nondependent function type
$\theta=\sigma_1\to\cdots\to\sigma_n\to\tau$, where $\tau$ is not a
function type, \sys introduces an uninterpreted sort $\fnsort{\theta}$
whose elements represent function values.
We refer to these uninterpreted sorts as $\fnsort{}$ sorts.
Its application symbol has signature
\[
  \appsym{\theta} :
  \fnsort{\theta}\times
  \sigma_1^\flat\times\cdots\times\sigma_n^\flat
  \to\tau^\flat.
\]
An application $f\,a_1\cdots a_n$ becomes
$\appsym{\theta}(f,a_1^\flat,\ldots,a_n^\flat)$,
where $f$ is a function symbol.
Quantification over a function therefore becomes ordinary
first-order quantification over the corresponding uninterpreted sort $\fnsort{}$.

\myparagraph{Lambdas and captured variables.}
For each lambda $\ell=\lambda x_1\cdots x_n.\,e$, \sys introduces a fresh uninterpreted symbol $\clossym{\ell}$ of sort $\fnsort{\theta_\ell}$.
Then, \sys generates an axiom about the closure application according to its definition, where $\bar y$ are variables captured by $\ell$:
\[
  \forall\bar y,\bar x.\quad
  \appsym{\theta}(\clossym{\ell}(\bar y),\bar x)
  =e^\flat.
\]
Here the variables $\bar x, \bar y$ range over their translated sorts, and the body is translated
under the same bindings. The closure symbol $\clossym{\ell}$ is uninterpreted, and its semantics are axiomatized according to its definition.

\myparagraph{Partial application and function equality.}
\sys applies $\eta$-expansion\footnote{Converting a function $f$ to its normal form, \IE $f=\lambda x.\, f x$}
to partial applications, transforming them into closures.
Then, it reuses the encoding for lambdas to encode the $\eta$-expanded terms.
When the query compares function values for equality, \sys additionally emits an
extensionality axiom for the corresponding $\fnsort{}$ sort.

\myparagraph{Guarded Sub-types.} For types represented by a subset of an SMT sort, such as
$\mathsf{Nat}$ within $\mathsf{Int}$, the encoding represents inhabitants of the Lean types and constrains the 
application results accordingly. For example, \sys reuses SMT solver's integer theory to represent
$\mathsf{Nat}$, and for each query involving a $\mathsf{Nat}$ $x$, the assertion will be guarded by $x\geq 0$.

\myparagraph{Encoding the running example.}
Let $\mathsf{B}$ denote the SMT sort of $\mathsf{BitVec}\;4$ and
$\theta=\mathsf{BitVec}\;4\to\mathsf{BitVec}\;4$.
The application symbol has signature
$\appsym{\theta}:\fnsort{\theta}\times\mathsf{B}\to\mathsf{B}$.
For the fixed mask $m$ in Fig.~\ref{fig:overview-example}, the closure on
line~\ref{line:overview-callback-body}, denoted by $\ell$, is represented by
the nullary closure symbol $\clossym{\ell}:\fnsort{\theta}$.
Combining the closure body with the custom translation
in Fig.~\ref{fig:overview-lowering} gives
\[
  \forall y:\mathsf{B}.\quad
  \appsym{\theta}(\clossym{\ell},y)=
  \mathsf{bvnand}(\mathsf{bvnand}(y,m),\mathsf{bvnand}(y,m)).
\]
The \leaninline{run} function on line~\ref{line:overview-callback} becomes an uninterpreted
symbol $\widehat{\mathsf{run}}:\fnsort{\theta}\times\mathsf{B}\to\mathsf{B}$.
Its hypothesis and the negated goal become, respectively, 
\[
  \begin{array}{l}
  \forall f:\fnsort{\theta},\,y:\mathsf{B}.\quad
    \widehat{\mathsf{run}}(f,y)=\appsym{\theta}(f,y),\\[3pt]
  \widehat{\mathsf{run}}(\clossym{\ell},x)\ne\mathsf{bvand}(x,m).
  \end{array}
\]
Instantiating the hypothesis with $\clossym{\ell}$ and $x$, and the closure equation
with $x$ completes the proof by contradiction.

\subsection{Soundness of Defunctionalization}
\label{sec:mechanized-soundness}

\begin{figure}[t!]
\centering
\begin{minipage}{\linewidth}
\small
\setlength{\abovedisplayskip}{3pt}
\setlength{\belowdisplayskip}{3pt}
\setlength{\abovedisplayshortskip}{1pt}
\setlength{\belowdisplayshortskip}{3pt}
\hrule
\vspace{3pt}
\begin{minipage}[t]{0.49\linewidth}
\setlength{\arraycolsep}{2pt}
\figurepart{Syntax of higher-order\\language (HO)}
\[
  \begin{array}{rcl}
  \sigma,\tau&::=&\booltype\mid\beta\mid\sigma\to\tau,\\[2pt]
  e&::=&x\mid c\mid\mathsf{true}\mid\mathsf{false}\\
   &&{}\mid\lambda x:\tau.\,e\mid e\,e\\
   &&{}\mid e=e\mid\neg e\mid e\circ e\\
   &&{}\mid Qx:\tau.\,e.
  \end{array}
\]
\end{minipage}\hfill
\begin{minipage}[t]{0.49\linewidth}
\setlength{\arraycolsep}{2pt}
\figurepart{Syntax of first-order\\language (FO)}
\[
  \begin{array}{rcl}
  s&::=&\booltype\mid\beta\mid\fnsort{\theta},\\[2pt]
  k&::=&\widehat c\mid\appsym{\theta}\mid\clossym{\ell},\\[2pt]
  t&::=&x\mid\mathsf{true}\mid\mathsf{false}\\
   &&{}\mid k(t_1,\ldots,t_n)\\
   &&{}\mid t=t\mid\neg t\mid t\circ t\\
   &&{}\mid Qx:s.\,t.
  \end{array}
\]
\end{minipage}
\[
  \circ\in\{\land,\lor,\to,\leftrightarrow\}\qquad
  Q\in\{\forall,\exists\}
\]
\hrule
\vspace{3pt}
\figurepart{Denotational semantics ($\sembrackets{\cdot}_I$)}
\vspace{2pt}
\begin{minipage}[t]{0.49\linewidth}
\figurepart{HO Sorts}
\[
  \begin{array}{@{}l@{}}
  \typedenote{\booltype}{M}=\proptype,\\[1pt]
  \typedenote{\beta}{M}=M(\beta)\;\mbox{s.t. }M(\beta)\ne\varnothing,\\[1pt]
  \typedenote{\sigma\to\tau}{M}=
    \typedenote{\sigma}{M}\to\typedenote{\tau}{M}.
  \end{array}
\]
\end{minipage}\hfill
\begin{minipage}[t]{0.49\linewidth}
\figurepart{FO Sorts}
\[
  \begin{array}{@{}l@{}}
  \typedenote{\booltype}{N}=\proptype,\\[1pt]
  \typedenote{\beta}{N}=N(\beta)\;\mbox{s.t. }N(\beta)\ne\varnothing,\\[1pt]
  \typedenote{\fnsort{\theta}}{N}=N(\fnsort{\theta})\;\mbox{s.t. }N(\fnsort{\theta})\ne\varnothing.
  \end{array}
\]
\end{minipage}
\par\vspace{2pt}
\figurepart{Constants, abstractions, and applications}
\[
  \begin{array}{rcl}
  \denote{c}{M}{\rho}&=&M(c),\\[1pt]
  \denote{\lambda x:\tau.\,e}{M}{\rho}&=&
    (\mathsf{fun}\;v\Rightarrow\denote{e}{M}{\rho[x\mapsto v]}),\\[1pt]
  \denote{e_1\,e_2}{M}{\rho}&=&
    \denote{e_1}{M}{\rho}(\denote{e_2}{M}{\rho}),\\[1pt]
  \denote{k(t_1,\ldots,t_n)}{N}{\rho}&=&
    N(k)(\denote{t_1}{N}{\rho},\ldots,\denote{t_n}{N}{\rho}).
  \end{array}
\]
\par\vspace{2pt}
\figurepart{Variables and logical connectives (HO and FO)}
\[
  \begin{array}{rcl@{\qquad}rcl}
  \denote{x}{I}{\rho}&=&\rho(x),&
    \denote{\mathsf{false}}{I}{\rho}&=&\mathsf{False},\\[1pt]
  \denote{\mathsf{true}}{I}{\rho}&=&\mathsf{True},&
    \denote{\neg u}{I}{\rho}&=&\neg\denote{u}{I}{\rho}.
  \end{array}
\]
\[
  \begin{array}{rcl}
  \denote{u=v}{I}{\rho}&=&
    (\denote{u}{I}{\rho}=\denote{v}{I}{\rho}),\\[1pt]
  \denote{u\circ v}{I}{\rho}&=&
    \denote{u}{I}{\rho}\circ\denote{v}{I}{\rho},\\[1pt]
  \denote{Qx:\alpha.\,u}{I}{\rho}&=&
    Qv:\typedenote{\alpha}{I}.\;\denote{u}{I}{\rho[x\mapsto v]}.
  \end{array}
\]
\hrule
\end{minipage}
\caption{HO syntax (left), FO syntax (right), and their denotational semantics.}
\label{fig:formal-languages}
\end{figure}

In this section, we describe the formalization of the defunctionalization pass and 
show the proof skeleton of its soundness. 
Full formalization in Lean is available in our artifact\footnote{\href{https://github.com/AD1024/lean-crush/tree/fc4191c16577088972b112c0b941a2c931281132/Crush/Metatheory}{Crush/Metatheory}}.
We first formalize the higher-order (HO) and first-order (FO) languages shown on the left
and right of Fig.~\ref{fig:formal-languages}, respectively.
Both languages have Boolean literals, equality, logical connectives, and quantifiers.
HO models the fragment of Lean supported by this transformation.
In HO, $\beta$ ranges over base (non-arrow) sorts,
$\sigma,\tau$ over
HO sorts, $e$ over terms, and $c$ over typed constants.
FO uses FO sorts $s$ (including the same base sorts $\beta$ as HO) and terms $t$, replacing each HO function sort $\theta$ by
$\fnsort{\theta}$ as described earlier.
The symbols $k$ are uncurried constant symbols
$\widehat c$, application symbols $\appsym{\theta}$, and closure symbols
$\clossym{\ell}$. Each has type
$s_1\times\cdots\times s_n\to s$ and is fully applied, with $n=0$
for nullary symbols.

\begin{definition}[Denotation and satisfaction]
\label{def:denotational-semantics}
Let $I$ be a model of either HO or FO.
We write $\typedenote{\alpha}{I}$ for the carrier of sort $\alpha$, i.e.,
the set of values interpreting $\alpha$ in model $I$.
We write $\denote{u}{I}{\rho}$ for the denotation of term $u$ in model $I$
under valuation $\rho$. We use $\varnothing$ to denote an empty valuation.
For a collection $T$ of closed formulas in the corresponding language, define
\[
  I\models T\quad\Longleftrightarrow\quad
  \forall\varphi\in T.\;\denote{\varphi}{I}{\varnothing}.
\]
\end{definition}

A higher-order model $M$ assigns nonempty carriers $M(\beta)$ to base sorts
and values $M(c)$ to constants; arrow sorts denote full function spaces.
A first-order model $N$ assigns arbitrary nonempty carriers to base sorts and
$\fnsort{}$ sorts, and interpretations $N(k)$ to symbols.
For denotational semantics in Fig.~\ref{fig:formal-languages}, we write $I$ for a model of either
language, $u,v$ for its terms, and $\alpha$ for a sort in that language. A valuation
$\rho$ assigns values to variables. Boolean terms denote Lean propositions ($\mathsf{Prop}$),
and quantifiers extend the valuation with a value from the interpretation of
their binder sort.
The sort encoding $\tau^\flat$ leaves Boolean and base sorts unchanged,
with $\theta^\flat=\fnsort{\theta}$ for an HO function sort $\theta$.

\begin{figure}[t!]
\centering
\begin{minipage}{\linewidth}
\small
\setlength{\abovedisplayskip}{5pt}
\setlength{\belowdisplayskip}{5pt}
\setlength{\abovedisplayshortskip}{3pt}
\setlength{\belowdisplayshortskip}{5pt}
\hrule
\vspace{5pt}
\[
  \theta=\sigma_1\to\cdots\to\sigma_n\to\tau,\quad
  \tau\in\{\booltype,\beta\}.
\]
\figurepart{Function application}
\[
  \frac{
    f:\theta\qquad a_i:\sigma_i\qquad
    \defunc(a_i)=(a_i^\flat,C_{a_i},E_{a_i})\quad(1\leq i\leq n)
  }{
    \defunc(f\,a_1\cdots a_n)=
    \bigl(\appsym{\theta}(f,\bar a^\flat),\;
      \textstyle\bigcup_i C_{a_i},\;
      \textstyle\bigcup_i E_{a_i}\bigr)
  }
\]
\par\vspace{5pt}
\figurepart{Closure construction}
\[
  \frac{
    \begin{array}{c}
      \ell=\lambda\bar x.\,e:\theta\qquad e:\tau\qquad
      \bar y=\mathsf{FV}(\ell)\\[3pt]
      \defunc(e)=(e^\flat,C_e,E_e)
    \end{array}
  }{
    \defunc(\ell)=
      (\clossym{\ell}(\bar y),\;C_e\cup\{K_\ell\},\;E_e)
  }
\]
\[
  K_\ell=(\forall\bar y,\bar x.\;
    \appsym{\theta}(\clossym{\ell}(\bar y),\bar x)=e^\flat).
\]
\par\vspace{5pt}
\figurepart{Function equality}
\[
  \frac{
    \begin{array}{c}
      u,v:\theta\qquad
      \defunc(u)=(u^\flat,C_u,E_u)\qquad
      \defunc(v)=(v^\flat,C_v,E_v)
    \end{array}
  }{
    \defunc(u=v)=
      (u^\flat=v^\flat,\;C_u\cup C_v,\;E_u\cup E_v\cup\{\extaxiom{\theta}\})
  }
\]
\[
  \extaxiom{\theta}=
    (\forall f,g:\fnsort{\theta}.\;
    (\forall\bar x.\;\appsym{\theta}(f,\bar x)
      =\appsym{\theta}(g,\bar x))\to f=g).
\]
\hrule
\end{minipage}
\caption{Defunctionalization (selected rules).}
\label{fig:defunctionalization-rules}
\end{figure}

\begin{definition}[Defunctionalization]
\label{def:defunctionalization}
Defunctionalization maps a well-typed higher-order term $e:\tau$ to
$\defunc(e)=(e^\flat,C_e,E_e)$, where $e^\flat:\tau^\flat$ is a
first-order term with first-order sort $\tau^\flat$ and $C_e,E_e$ are sets of closed closure equations and
extensionality axioms, respectively. Core definitions are given in
Fig.~\ref{fig:defunctionalization-rules}. Complete definitions are shown in
Appendix~\ref{app:formal-languages}.
\end{definition}

In Fig.~\ref{fig:defunctionalization-rules}, $\theta$ has a complete
argument list $\bar\sigma$ and non-function result $\tau$.
Each $a_i^\flat$ is the first-order term returned by the recursive call
$\defunc(a_i)$, including defunctionalization when $a_i$ has function type.
Application of a function variable $f:\theta$ passes these terms to
$\appsym{\theta}$ and accumulates their generated axioms.

For a lambda $\ell=\lambda\bar x.\,e$, the closure symbol takes its free
variables $\bar y=\mathsf{FV}(\ell)$, and $K_\ell$ specifies
its behavior using the recursively translated non-function body $e$.
Variables of $\fnsort{}$ sorts, constants, and partial applications are $\eta$-expanded over their remaining arguments and use the same closure construction.
Appendix~\ref{app:formal-languages} gives the generalized rule,
including lambdas whose bodies still have function type.
Equality at function type $\theta$ additionally contributes
$\extaxiom{\theta}$ that encodes functional extensionality.
Logical connectives and quantifiers
preserve their operators, replacing quantifications to the corresponding first-order sorts.

\begin{definition}[First-order encoding]
\label{def:complete-encoding}
For a finite collection $T$ of closed higher-order formulas, let
$(\varphi^\flat,C_\varphi,E_\varphi)=\defunc(\varphi)$ and define
\[
  \theoryencoding(T)=\bigcup_{\varphi\in T}
    \bigl(C_\varphi\cup E_\varphi\cup\{\varphi^\flat\}\bigr).
\]
\end{definition}


\begin{definition}[Canonical model extension]
\label{def:canonical-extension}
For each lambda $\ell$ with free variables $\bar y=(y_1,\ldots,y_m)$,
let $\rho_{\ell,\bar v}$ be the valuation defined by
$\rho_{\ell,\bar v}(y_i)=v_i$ for $1\leq i\leq m$.
The extension $M^\star$ of $M$ interprets the generated symbols by
\[
  \begin{array}{rcl}
  \typedenote{\fnsort{\theta}}{M^\star} &=& \typedenote{\theta}{M},\\[2pt]
  M^\star(\appsym{\theta})(f,\bar a) &=& f\,a_1\cdots a_n,\\[2pt]
  M^\star(\clossym{\ell})(\bar v)
    &=& \denote{\ell}{M}{\rho_{\ell,\bar v}}.
  \end{array}
\]
Since corresponding HO and FO sorts have the same carriers in $M$ and
$M^\star$, we use $\rho$ for both valuations.
\end{definition}

\begin{lemma}[Preservation of denotations]
\label{lem:denotation-preservation}
For every well-typed higher-order term $e$, model $M$, and valuation $\rho$,
with $M^\star$ as in Definition~\ref{def:canonical-extension},
\[
  \denote{e^\flat}{M^\star}{\rho}=\denote{e}{M}{\rho}.
\]
\end{lemma}
\begin{proofsketch}
We proceed by induction on typed terms. In the lambda case, we use the fact that denotation
depends only on free variables, so reconstructing the captured values suffices.
The full proof is shown in Appendix~\ref{app:proof-denotation}.
\end{proofsketch}

\begin{lemma}[Validity of generated axioms]
\label{lem:generated-axioms}
For every model $M$ and well-typed term $e$,
\[
  M^\star\models C_e\cup E_e.
\]
\end{lemma}
\begin{proofsketch}
We establish closure equations using Definition~\ref{def:canonical-extension}
and Lemma~\ref{lem:denotation-preservation}.
In $M^\star$, the $\fnsort{}$ sorts are interpreted as function spaces and the application symbols as function application, so the extensionality axioms follow from function extensionality. The full proof appears in
Appendix~\ref{app:proof-axioms}.
\end{proofsketch}

We combine these properties to obtain the following result.

\begin{theorem}[Soundness]
\label{thm:defunctionalization-soundness}
For every finite collection $T$ of closed higher-order formulas, if
$\theoryencoding(T)$ is unsatisfiable, then $T$ is unsatisfiable.
\end{theorem}
\begin{proofsketch}
We argue by contradiction. If $M\models T$, Lemma~\ref{lem:denotation-preservation} makes every translated
formula true in $M^\star$, while Lemma~\ref{lem:generated-axioms} supplies all auxiliary axioms. Thus $M^\star\models\theoryencoding(T)$, contradicting its unsatisfiability. The full proof is shown in Appendix~\ref{app:proof-soundness}.
\end{proofsketch}


\section{Proof Reconstruction}
\label{sec:reconstruction}

For a goal $G$ in Lean, \sys emits the translated hypotheses and $\neg G$ to the
solver. When the solver reports SAT, it provides a model of the encoded first-order formula, possibly spurious if an opaque
definition was not unfolded; when it reports UNKNOWN, the Lean goal is left unresolved. When it reports UNSAT, the solver may provide
an UNSAT core, 
a proof script, or both.
In the UNSAT case, \sys attempts to reconstruct a kernel-checked proof through two approaches: use the UNSAT core to select premises
and search the proof via tactics in Lean, or replay the solver's inference steps (proof scripts), translating them into Lean proof terms.


\subsection{Core-Guided Reconstruction}
\label{sec:core-reconstruction}

An UNSAT core contains assertions sufficient for a contradiction, usually much smaller than the input. \sys
maps them back to Lean hypotheses and searches for a proof by contradiction
from the hypotheses and the negated goal leveraging existing proof automation tactics (\EG \leaninline{omega} or \leaninline{grind}).
Appendix~\ref{app:reconstruction-procedures} provides details for the search procedure.
To extend search procedure, users may mark a theorem
as a reconstruction hint via \leaninline{@[crush_reconstruct]}. When the current goal matches the conclusion of a marked theorem, it is applied to the goal, and \sys proves the premises recursively.

\subsection{Extensible Solver Inference Replay}
\label{sec:alethe-replay}

Alternatively, \sys can replay the proof script generated by SMT solvers to reconstruct kernel-checked Lean proofs.
A proof script is a list of inference steps, where each inference step records a conclusion, a rule, and any premises or arguments used by the rule to derive the conclusion.
\sys relies on rewrite rules to translate terms and inference steps in proof scripts into Lean term and Lean proofs.
Users may extend those rewrite rules via \sys's extensible and type-aware interfaces.
Appendix~\ref{app:reconstruction-procedures} provides the details of the replay procedure.

\myparagraph{Term decoding.}
Users may specify translations from SMT operators in proof scripts into Lean terms via \leaninline{register_crush_replay term} command.
The translated Lean terms are further type-checked by the Lean kernel.
As an example, Fig.~\ref{fig:overview-replay} defines a handler for the SMT operator $\mathsf{bvnand}$, which requires the arguments to have Lean type $\mathsf{BitVec}$ and translates the expression into a complemented conjunction.

\begin{figure}[t!]
  \centering
  \begin{minipage}[t]{0.47\linewidth}
  \figurepart{Alethe inference}
  \begin{lstlisting}[language=Alethe,frame=single,numbers=none,xrightmargin=4pt,
    aboveskip=4pt,belowskip=0pt]
(step s
  (cl (= (bvnot (bvnot a)) a))
  :rule rare_rewrite
  :args ("bv-not-idemp" a))
\end{lstlisting}
  \end{minipage}\hfill
  \begin{minipage}[t]{0.50\linewidth}
  \figurepart{Lean inference handler}
  \begin{lstlisting}[frame=single,numbers=left,firstnumber=1,
    numberstyle=\tiny\color{gray},xleftmargin=1.5em,xrightmargin=4pt,
    aboveskip=4pt,belowskip=0pt,escapeinside={(*@}{@*)}]
register_crush_replay rule <<(*@\label{line:replay-rule-register}@*)
  (rare_rewrite "bv-not-idemp"(*@\label{line:replay-rule-pattern}@*)
    (term x : BitVec 4)) =>(*@\label{line:replay-rule-capture}@*)
  by exact BitVec.not_not (b := x)(*@\label{line:replay-rule-proof}@*)
>>
\end{lstlisting}
  \end{minipage}
  \caption{An inference from the running example's certificate and a replay handler.}
  \label{fig:reconstruction-rule}
\end{figure}

\myparagraph{Solver inference replay.}
Users may define translations of inference steps into a Lean proof with \leaninline{register_crush_replay rule} command.
For example, the left of Fig.~\ref{fig:reconstruction-rule} shows a use of \leaninline{"bv-not-idemp"} rewrite rule in the solver's inference step, which rewrites a SMT term $\mathsf{bvnot}(\mathsf{bvnot}(a))$ into $a$ and asserts its equality.
The right of Fig.~\ref{fig:reconstruction-rule} registers a rule to translate \leaninline{"bv-not-idemp"} in the proof script into a Lean proof of $\mathord{\sim}(\mathord{\sim}x)=x$ (where $\sim$ denotes bitwise complement in Lean), using the existing lemma \leaninline{BitVec.not_not}.

\myparagraph{Current limitations.}
The proof replay procedures in \sys 
rely on the SMT solver's ability to produce a proof script. For CVC5, its Alethe proof generation does not cover all theories it can solve.
Moreover, the built-in rules in
\sys's replay procedures 
may not cover every SMT encoding or inference step. 
In practice, 
these gaps can be 
filled by additional user-supplied rules.



\section{Evaluation}
\label{sec:evaluation}

We compare \sys against related proof automation frameworks on metrics for 
verification coverage, runtime, and proof reconstruction. We first compare overall coverage and runtime (\S\ref{sec:eval-coverage}) for each benchmark suite, and then compare \sys pair-wise with each baseline framework on the VCs solved by both (\S\ref{sec:eval-head-to-head}).
Finally, we examine how core-guided reconstruction and solver inference replay contribute to obtaining kernel-checked proofs (\S\ref{sec:eval-reconstruction}).

\begin{figure}[t!]
  \centering
  \includegraphics[width=0.495\linewidth]{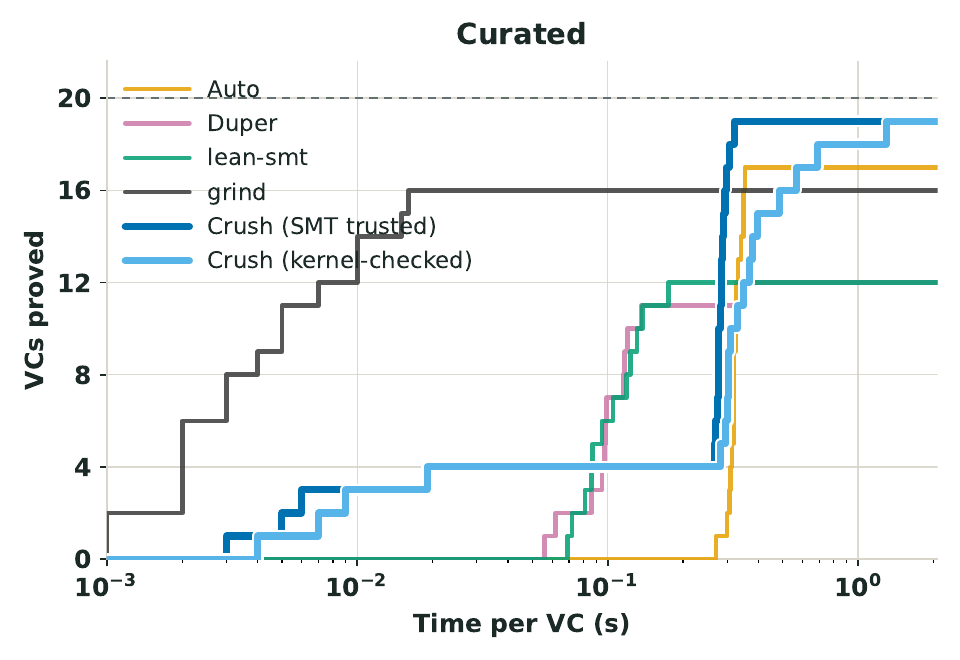}\hfill
  \includegraphics[width=0.495\linewidth]{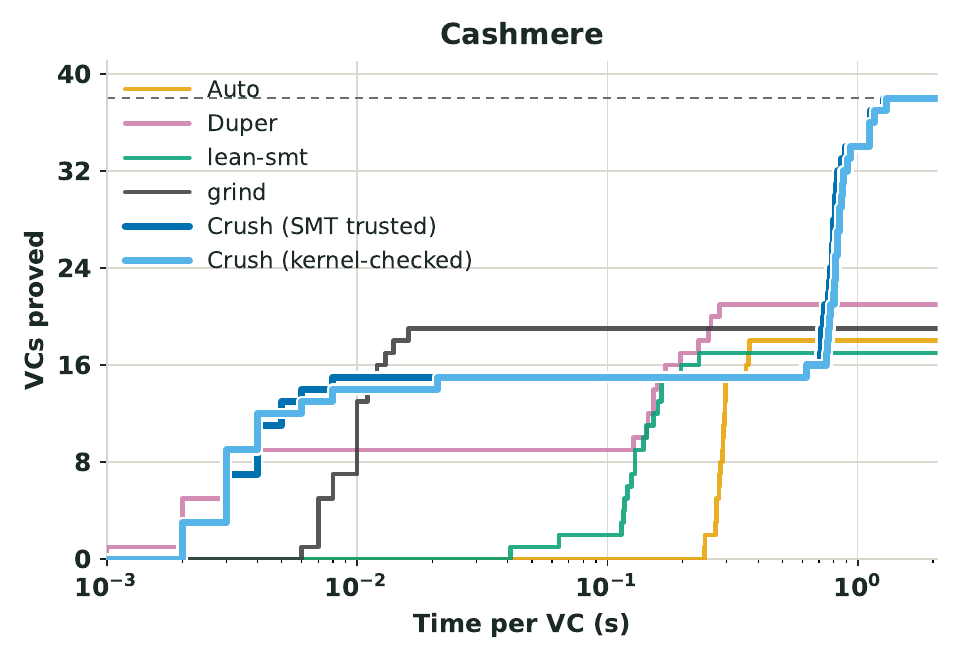}
  \par\smallskip
  \includegraphics[width=0.495\linewidth]{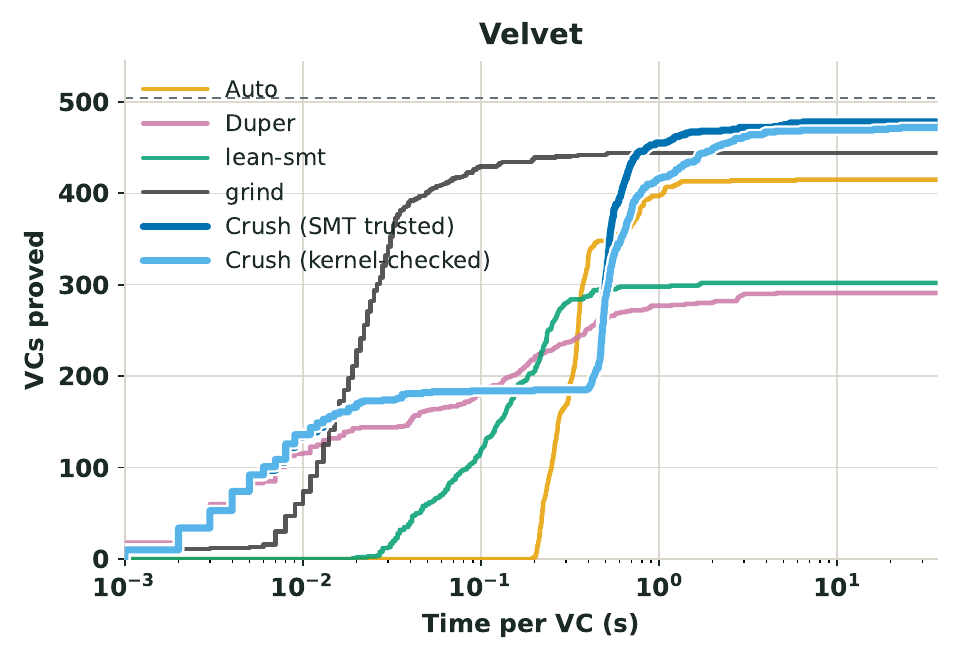}\hfill
  \includegraphics[width=0.495\linewidth]{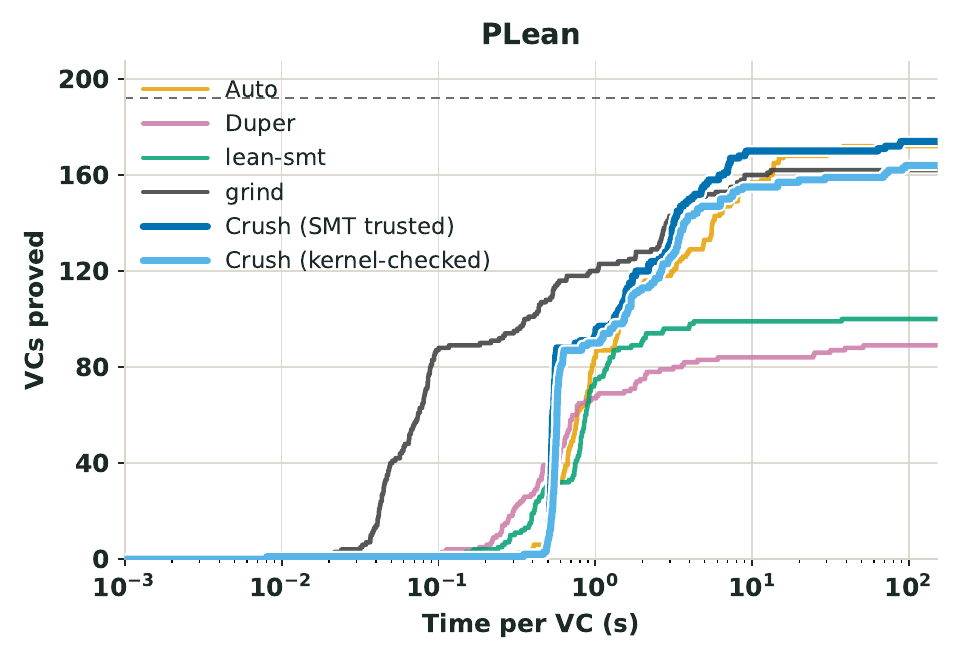}
  \caption{Verification coverage against time per VC.}
  \label{fig:eval-coverage}
\end{figure}

\myparagraph{Experimental setup.}
The benchmark comprises 754 VCs from four suites: Curated (20),
Cashmere~\cite{Cashmere-github} (38), Velvet~\cite{velvet} (504), and PLean~\cite{plean} (192).
We compare \sys with lean-auto (\leanauto)~\cite{lean-auto},
Duper~\cite{duper}, \leansmt~\cite{lean-smt}, and Grind (\grind)~\cite{grind-tactic}.
Curated is a standalone collection of Lean goals developed by us, covering higher-order functions, datatypes, arithmetic, bitvectors,
and strings.
Cashmere, Velvet and PLean use Loom~\cite{loom} in their VC generation. 
Cashmere and Velvet are case studies of Loom, covering VCs from program verification. 
PLean covers verification of distributed protocols written in P~\cite{ankush2013p} language embedded in Lean.
We plug \sys and the baselines into Loom as SMT backends to discharge the generated VCs.
We evaluate \sys both with SMT trusted (\IE without proof reconstruction) and with kernel-checked proof reconstruction.
The former trusts the solver's verdict and, like \leanauto, uses an axiom to close the goal.
The latter requires a kernel-checked proof of the original goal and combines core-guided reconstruction with replay of the solver's inference steps recorded in the Alethe proof format (\emph{Alethe replay}).
In this case, \sys is considered successful if it produces a kernel-checked proof for the original goal.

All experiments ran on an Apple M1 Pro CPU with 32\,GB of memory.
We use cvc5~1.3.4 as the SMT backend.
We use Duper revision \texttt{ca7c5862}, \leansmt revision
\texttt{e5025665}, and Lean~4.32.2 (except for \leanauto on Cashmere, Velvet, and PLean, for which we retain Lean~4.24.0 to preserve the existing Loom--lean-auto integration).
We set a limit of 5 seconds on each solver call and a limit of 1M Lean heartbeats on each tactic invocation.
We measure the time for each tactic on each VC, including translation, solver execution, and proof reconstruction.

\subsection{Verification Coverage and Runtime}
\label{sec:eval-coverage}

Fig.~\ref{fig:eval-coverage} shows the number of VCs discharged within each
time threshold per VC.
Here, the 
x-axis shows the time threshold in seconds for each VC on a log scale; the y-axis shows the number of VCs discharged within that time threshold.
The dashed horizontal line marks the total number of VCs in each benchmark suite.

Across the four benchmarks, \sys discharges 709 VCs (94.0\%) when trusting the solver's results,
and 693 (91.9\%) when enabling proof reconstruction. 
Among the compared baselines,
\grind achieves the highest coverage (641 VCs), followed by \leanauto (622), \leansmt (431), and Duper (413).
\emph{Notably, \sys achieves the highest overall coverage even when kernel-checked proofs are required.}
The runtime curves show a tradeoff: \grind discharges many VCs at shorter time thresholds, while \sys achieves higher overall coverage. The advantage also varies by benchmark suites. On PLean, \leanauto discharges more VCs than kernel-checked \sys (172 versus 164), although trusted \sys discharges 174.

\subsection{Pairwise Comparison with \sys}
\label{sec:eval-head-to-head}

\begin{table}[t!]
\centering
\caption{Head-to-head comparison with \sys using proof reconstruction.
Times are mean $\pm$ standard deviation across VCs discharged by both tools, after averaging repetitions per VC. 
}
\label{tab:eval-head-to-head}
\footnotesize
\setlength{\tabcolsep}{3pt}
\renewcommand{\arraystretch}{1.12}
\begin{tabular}{@{}llrrrrr@{}}
\hline
& & \multicolumn{3}{c}{\textbf{\# VCs Discharged}}
& \multicolumn{2}{c}{\textbf{Time on Both (ms)}} \\
\cline{3-5}\cline{6-7}
\noalign{\vskip 3pt}
\textbf{Benchmark} & \textbf{Baseline}
& \shortstack{\sys\\\textbf{only}} & \shortstack{\textbf{Baseline}\\\textbf{only}}
& \textbf{Both} & \sys & \textbf{Baseline} \\
\hline
\multirow{4}{*}{\shortstack[l]{Curated\\(20 VCs)}}
& \leanauto & 2 & 0 & 17 & $365.0_{\pm 298.1}$ & $322.2_{\pm 20.1}$ \\
& Duper & 8 & 1 & 11 & $313.1_{\pm 191.1}$ & $98.3_{\pm 24.1}$ \\
& \leansmt & 7 & 0 & 12 & $404.2_{\pm 338.6}$ & $106.8_{\pm 31.5}$ \\
& \grind & 3 & 0 & 16 & $371.4_{\pm 309.3}$ & $5.5_{\pm 4.8}$ \\
\hline
\multirow{4}{*}{\shortstack[l]{Cashmere~\cite{Cashmere-github}\\(38 VCs)}}
& \leanauto & 20 & 0 & 18 & $235.9_{\pm 395.8}$ & $294.0_{\pm 36.0}$ \\
& Duper & 17 & 0 & 21 & $252.9_{\pm 407.1}$ & $109.3_{\pm 102.6}$ \\
& \leansmt & 21 & 0 & 17 & $156.5_{\pm 352.4}$ & $135.6_{\pm 44.2}$ \\
& \grind & 19 & 0 & 19 & $223.6_{\pm 388.4}$ & $9.8_{\pm 2.7}$ \\
\hline
\multirow{4}{*}{\shortstack[l]{Velvet~\cite{velvet}\\(504 VCs)}}
& \leanauto & 70 & 13 & 402 & $598.5_{\pm 1400.7}$ & $372.5_{\pm 215.2}$ \\
& Duper & 184 & 3 & 288 & $287.0_{\pm 483.2}$ & $220.3_{\pm 517.4}$ \\
& \leansmt & 180 & 10 & 292 & $611.4_{\pm 1584.3}$ & $142.1_{\pm 95.9}$ \\
& \grind & 37 & 9 & 435 & $460.8_{\pm 595.6}$ & $26.1_{\pm 26.8}$ \\
\hline
\multirow{4}{*}{\shortstack[l]{PLean~\cite{plean}\\(192 VCs)}}
& \leanauto & 2 & 10 & 162 & $4517.8_{\pm 14086.1}$ & $3380.3_{\pm 5970.0}$ \\
& Duper & 75 & 0 & 89 & $644.6_{\pm 284.3}$ & $2755.4_{\pm 8169.2}$ \\
& \leansmt & 65 & 1 & 99 & $1611.3_{\pm 9096.7}$ & $1303.1_{\pm 3731.4}$ \\
& \grind & 2 & 0 & 162 & $4493.5_{\pm 14089.6}$ & $1206.5_{\pm 2327.6}$ \\
\hline
\end{tabular}
\end{table}

Table~\ref{tab:eval-head-to-head} compares \sys (with proof reconstruction)
against each baseline on the same VCs.
Runtimes are reported as mean $\pm$ standard deviation over VCs discharged by both tools.
\sys discharges more VCs in all pairs except against \leanauto on PLean.
The trusted-mode comparison is in Table~\ref{tab:eval-head-to-head-trusted},  Appendix~\ref{app:trusted-comparison}.

The pairwise results show that \sys can prove many obligations missed by the
baselines. Across the four suites, \sys proves 61 VCs that \grind cannot
close, while \grind proves 9 that \sys cannot close. On Velvet, \sys
proves 184 VCs missed by Duper and 180 missed by \leansmt; these baselines
prove only 3 and 10 VCs, respectively, that \sys misses.
On VCs solved by both tools, \sys has a higher mean runtime in most comparisons,
and \grind is faster on all four suites. The tradeoff varies by workload:
on PLean, \sys is slower on average than \leansmt (1611.3 versus 1303.1\,ms)
on the 99 VCs, but it also proves 65 VCs that \leansmt misses.

We observe that reconstruction failures mainly arise when cvc5 cannot produce a proof script and core-guided reconstruction also fails to recover a
proof. 
We examine these cases in more detail in the next section.

\subsection{Proof Reconstruction}
\label{sec:eval-reconstruction}

We evaluate how combining Alethe replay with core-guided reconstruction
affects proof coverage, examine the remaining failure modes, and report
the manual effort needed to define replay handlers.

\begin{figure}[t!]
  \centering
  \includegraphics[width=0.495\linewidth]{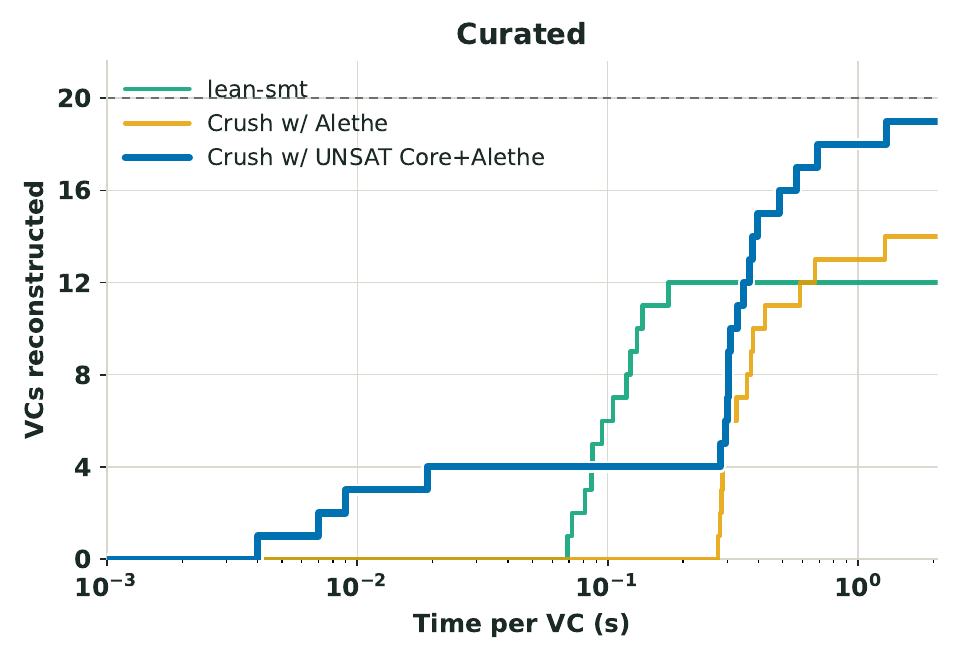}\hfill
  \includegraphics[width=0.495\linewidth]{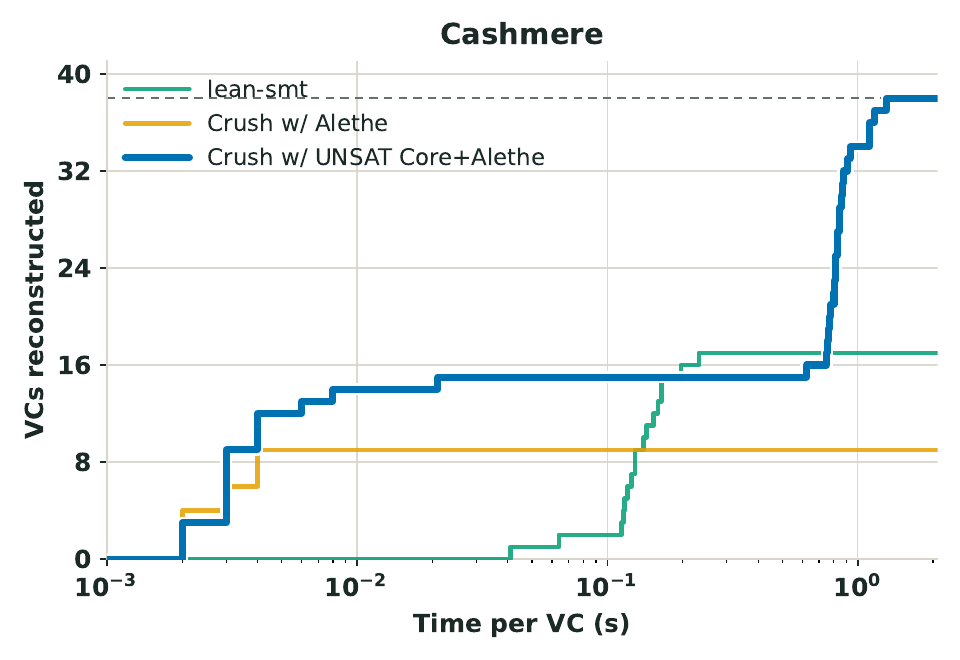}
  \par\smallskip
  \includegraphics[width=0.495\linewidth]{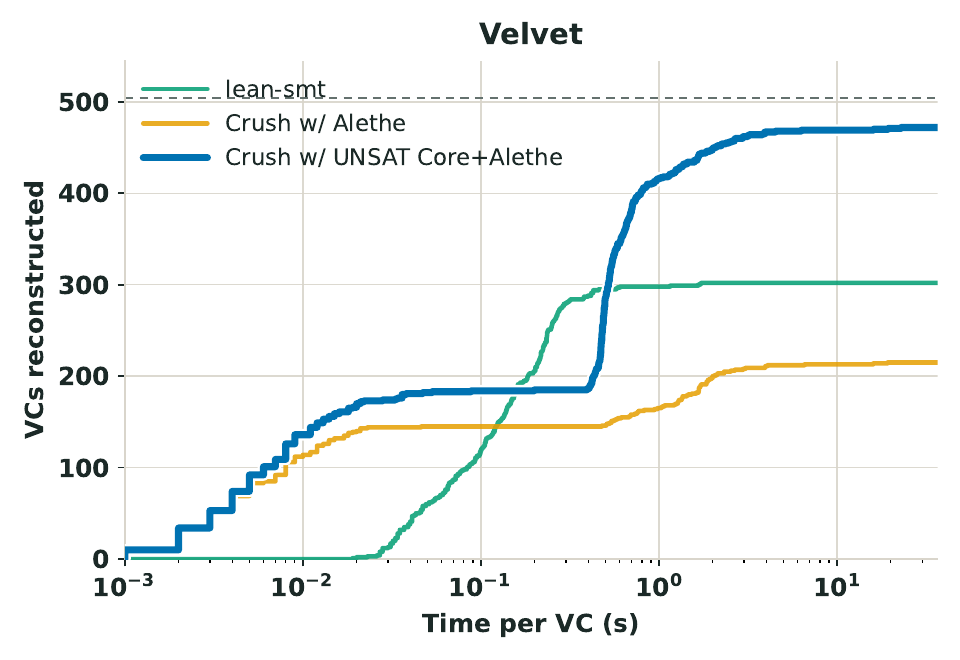}\hfill
  \includegraphics[width=0.495\linewidth]{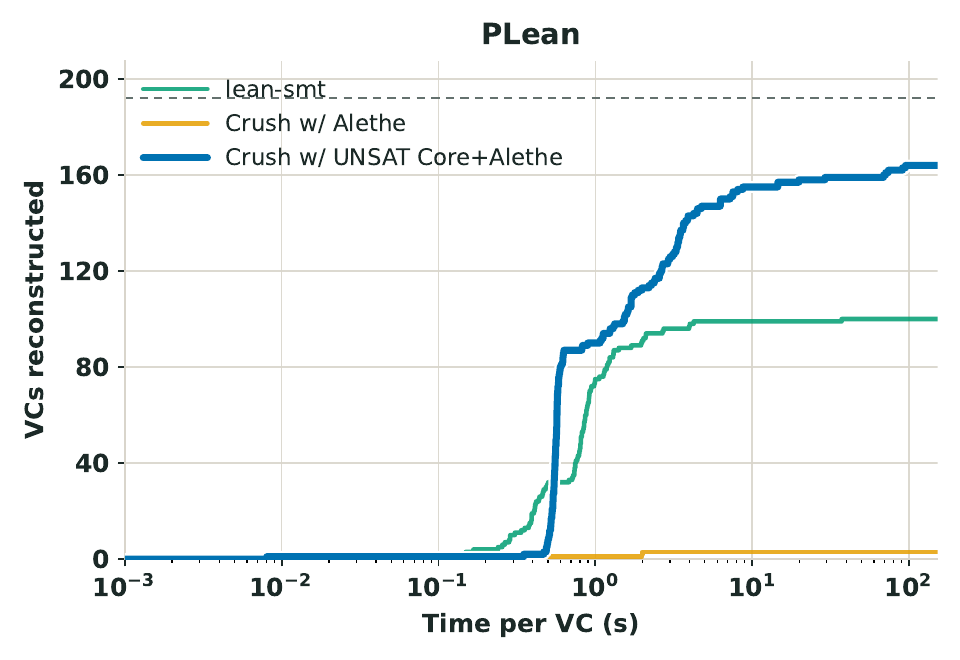}
  \caption{Kernel-checked proof reconstruction coverage over time per VC.}
  \label{fig:eval-reconstruction}
\end{figure}

\begin{table}[t!]
  \centering
  \caption{Proof reconstruction failures among VCs discharged by SMT.}
  \label{tab:eval-reconstruction-failures}
  \small
  \setlength{\tabcolsep}{5pt}
  \renewcommand{\arraystretch}{1.1}
  \begin{tabular}{@{}llr@{}}
    \hline
    \textbf{Method} & \textbf{Failure mode} & \textbf{VCs} \\
    \hline
    \multirow{2}{*}{Alethe replay}
      & Certificate generation & 431 \\
      & Term decoding & 1 \\
    \cline{2-3}
      & \textbf{Total} & \textbf{432 / 503} \\
    \hline
    \multirow{3}{*}{Combined}
      & Certificate generation and core-guided reconstruction & 11 \\
      & Lean heartbeat limit & 2 \\
      & Tactic failure during replay & 3 \\
    \cline{2-3}
      & \textbf{Total} & \textbf{16 / 503} \\
    \hline
  \end{tabular}
\end{table}

\myparagraph{Reconstruction coverage.}
Fig.~\ref{fig:eval-reconstruction} shows the number of VCs that have
proofs reconstructed by \sys and \leansmt in each case study, with the same x- and y-axes as in Fig.~\ref{fig:eval-coverage}.
In total, \sys closes 693 VCs when combining Alethe replay with core-guided reconstruction, compared with 241 using Alethe replay alone and 431 for \leansmt. 
Adding core-guided reconstruction therefore enables \sys to close more VCs.

\myparagraph{Failure analysis for Alethe replay.}
We further categorized the failure modes to identify 
specific limitations in the Alethe replay, 
and which limitations were overcome by core-guided reconstruction. 
Solver certificate generation caused 431 of the replay failures among the 503 VCs discharged by SMT.
Although cvc5 concludes unsatisfiability, it does not fully support proof generation in Alethe.
For example, it currently cannot generate a proof for the theory of arrays,
which is used in the Cashmere and Velvet suites.
Here, core-guided reconstruction often allows \sys to obtain kernel-checked proofs even when the solver's proof-generation support is insufficient.
The other replay failure stems from a \aletheinline{choice} expression in the certificate, which
chooses an element from a set.
\sys does not currently support decoding this expression, since it would require synthesizing a proof of non-emptiness of the set from which the element is chosen. 

In 11 VCs, certificate generation fails and core-guided reconstruction also fails to find a proof: 4 occur in Velvet and 7 in PLean.
The remaining failures are due to exhausted heartbeat budgets and tactic failures, caused by insufficient lemmas
for the built-in proof search procedures during replay.

\begin{table}[t!]
  \centering
  \caption{Implementation size of built-in Alethe replay extensions.}
  \label{tab:eval-replay-effort}
  \small
  \setlength{\tabcolsep}{5pt}
  \renewcommand{\arraystretch}{1.1}
  \begin{tabular}{@{}lrrrr@{}}
    \hline
    \textbf{Category} & \textbf{Handlers} & \textbf{Handler LoC}
      & \textbf{Lemma LoC} & \textbf{Total} \\
    \hline
    Assumption normalization & 2 & 16 & 34 & 50 \\
    Logic and equality & 10 & 47 & 10 & 57 \\
    Quantifiers & 1 & 3 & 22 & 25 \\
    Arithmetic & 17 & 68 & 311 & 379 \\
    Bitvectors & 13 & 72 & 337 & 409 \\
    Strings & 12 & 50 & 83 & 133 \\
    \hline
    \textbf{Total} & \textbf{55} & \textbf{256} & \textbf{797} & \textbf{1,053} \\
    \hline
  \end{tabular}
\end{table}

\myparagraph{Replay implementation effort.}
The built-in replay handler library covers assumption normalization, logic and equality,
quantifiers, arithmetic, bitvectors, and strings through the extension
interface (\S\ref{sec:alethe-replay}). The 55 manually defined handlers take 256 LoC, with another
797 LoC of supporting lemmas, for a total of 1,053 LoC 
(Table~\ref{tab:eval-replay-effort}).

\section{Related Work}
\label{sec:related-work}

\paragraph{ATP Integration in Proof Assistants.}
CoqHammer~\cite{coqhammer} translates Rocq goals through CIC$_0$ into untyped first-order logic.
Although there is a soundness proof for CoqHammer's core encoding, its implementation uses approximations that can compromise soundness, including omitted type guards in lambda-lifting equations. It ensures soundness by reconstructing kernel-checked Rocq proofs using ATP-selected premises.
\sys instead mechanizes the soundness of defunctionalization for its supported higher-order fragment.

SMTCoq~\cite{smtcoq} reifies Rocq goals into SMT formulas and validates solver certificates using a verified checker. Its modular architecture supports additional solvers and theories; it provides extensible preprocessing through Trakt~\cite{Blot_2023}.
It targets universal goals over supported SMT theories.
In contrast, \sys supports higher-order terms and existential quantifications, and natively provides extensible translation and proof replay.
Its handlers construct kernel-checked Lean proofs without requiring
separate soundness proofs for their implementations.


\paragraph{Automation Tactics in Lean.}
Several prior works attempted to close the gap between Lean and SMT solvers for better proof automation.
lean-auto~\cite{lean-auto} uses a monomorphization pass to translate dependent type theory into higher-order logic, but its SMT backend supports a restricted first-order fragment without SMT proof reconstruction or custom encoding interfaces.
lean-smt~\cite{lean-smt} uses lean-auto for monomorphization and provides extensible SMT translation, term decoding, and inference replay.
Its translation similarly requires a supported first-order fragment.
LeanHammer~\cite{leanhammer} combines premise selection, lean-auto's higher-order translation, Zipperposition~\cite{zipperposition} proof search, and Duper~\cite{duper} reconstruction from selected premises.
More recent versions also integrate lean-smt, hence inheriting its SMT translation and replay capabilities.

\sys is also broadly related to tactics that are natively implemented in Lean.
Grind~\cite{grind-tactic} implements several SMT solving procedures in Lean to perform proof search; Duper~\cite{duper} implements proof-producing superposition, similar to Isabelle's Metis~\cite{isabelle-metis}; Aesop~\cite{aesop-tactic} implements white-box proof search in Lean.


\section{Conclusion and Future Work}
\label{sec:conclusion}

We present \sys, an extensible proof automation framework for Lean.
It provides extensible translation interfaces for encoding Lean proof obligations into first-order SMT and leverages defunctionalization to support higher-order constructs. 
We have also provided a soundness proof of this translation procedure for the supported fragment of Lean.
\sys reconstructs proofs from solver certificates and supports proof search guided by unsat cores and replaying inference steps produced by proof-producing SMT solvers. 
On 4 benchmark suites with 754 verification conditions in Lean, \sys solves 709 (94.0\%) when trusting the SMT solver and 693 (91.9\%) when proof reconstruction is required, outperforming all the baseline approaches. 

\paragraph{Future work.} We aim to support dependent function types and extend the soundness formalization beyond defunctionalization to Lean-to-SMT-LIB translation.
For proof reconstruction, we plan to integrate tactics such as Duper\cite{duper}, support CPC ~\cite{cpc} certificates, and synthesize non-emptiness proofs for \texttt{choice} expressions in solver proof certificates.
We also aim to improve the feedback to the user when the SMT solver reports SAT, for example, by distinguishing true and spurious counter-examples.


\newpage
\section*{Data Availability}

The implementation of \sys, its mechanized proofs, benchmark data, and
scripts for reproducing the experiments are publicly available at
\url{https://github.com/AD1024/lean-crush/tree/tacas-ae}.
\bibliographystyle{splncs04}
\bibliography{references}

\clearpage
\appendix
\renewcommand{\theHsection}{appendix.\Alph{section}}
\renewcommand{\theHsubsection}{\theHsection.\arabic{subsection}}

\begin{center}
  \huge\textbf{Appendix}
\end{center}
\DoToC
\newpage

\section{Details of Pre-processing passes}
\label{app:query-preparation}

Given a Lean proof goal, \sys first constructs a set of propositions and pre-processes them before translating (defunctionalizing) them.

\subsection{Fact Collection and Normalization}

\sys collects the negated goal, local propositional hypotheses, and lemmas
supplied explicitly or selected from the environment. It retains their origins
so an unsat core can later be mapped back to the corresponding Lean facts.

\sys further applies normalization on the collected propositions.
If they use definitions that are marked with \leaninline{crush_unfold}, \sys will rewrite their applications using the definitions' equation lemmas and also expose these (quantified) equations as facts in future passes.
We note that such rewriting is performed by tactics that produce equality proofs checked by Lean kernel.
Moreover, \sys will not expose equations of unused definitions, preventing introduction of quantifiers and instantiations that are unproductive.

\emph{Example.}
Suppose $\mathit{double}(x)$ is defined as $x+x$ and marked with \leaninline{crush_unfold}.
For a goal $0 \leq \mathit{double}(x)$ under the hypothesis
$0 \leq x$, fact collection produces
\[
  0 \leq x
  \qquad\mbox{and}\qquad
  \neg\bigl(0 \leq \mathit{double}(x)\bigr).
\]
Because the marked definition occurs in the negated goal, normalization rewrites
the latter fact to $\neg(0 \leq x+x)$. The recorded equality proof connects
this normalized contradiction back to the original goal.

\subsection{Lemma Monomorphization}

Lean lemmas may be polymorphic, \IE quantifying over types, whereas SMT symbols have concrete sorts.
Translating the type variables of such a lemma into some uninterpreted sorts will make SMT solvers unable to use it for a concrete sort, \EG $\mathsf{Int}$.
\sys therefore collects the concrete data types in the proof obligation and specializes the type parameters of polymorphic lemmas at those types, maybe iteratively.
Every instance is obtained by applying the original Lean proof.
The numbers of instances and iterations are bounded by user-configurable options, which could reduce the completeness but still preserve soundness.
We note that function types are handled separately by the defunctionalization procedure in \S\ref{sec:higher-order}.

\emph{Example.}
Consider a polymorphic function $f$ and a hypothesis
\[
  h : \forall (\alpha : \mathsf{Type})\,(xs : \mathsf{List}\;\alpha),
      f\;\alpha\;xs = 0.
\]
If the goal is $f\;\mathsf{Int}\;ys=0$, where
$ys : \mathsf{List}\;\mathsf{Int}$, the occurrences of
$\mathsf{Int}$ and $\mathsf{List}\;\mathsf{Int}$ direct \sys to apply
$h$ to $\mathsf{Int}$. The resulting fact
\[
  h_{\mathsf{Int}} : \forall (xs : \mathsf{List}\;\mathsf{Int}),
      f\;\mathsf{Int}\;xs = 0
\]
uses the same concrete SMT symbols as the goal while retaining $h\;\mathsf{Int}$ in its Lean proof.

\subsection{Bounded Quantifier Instantiation}

Bounded quantifier instantiation supplies solver-side E-matching with a finite set of relevant ground
instances. E-matching requires a suitable ground term before instantiating a quantified
proposition, so it can miss a forward chain in which one instance produces the
trigger for the next. Therefore, \sys applies active bounded instantiation by matching quantified
lemmas against ground terms and proposition occurrences already in the query,
including newly generated instances.
Note that every instance is constructed by applying its quantified Lean proof to the selected arguments.

\sys takes a "two-shot" approach after quantifier instantiation. It first emits an SMT query without the quantified lemmas, avoiding unbounded solver-side instantiation. It closes the goal if this query is already UNSAT. If the query is SAT or unknown, \sys retries with the complete set of propositions including quantified lemmas.

\emph{Example.}
Let the selected lemmas and goal be
\[
\begin{array}{ll}
  \mathit{step} &: \forall x,\ R\;x\;(\mathit{next}\;x), \\
  \mathit{lift} &: \forall x\,y,\ R\;x\;y \rightarrow P\;y, \\
  \mathit{goal} &: \exists y,\ R\;a\;y \land P\;y.
\end{array}
\]
The ground term $a$ first instantiates $\mathit{step}$, producing
$R\;a\;(\mathit{next}\;a)$. That new occurrence fixes both arguments of
$\mathit{lift}$, producing $P\;(\mathit{next}\;a)$ in the next round.
The solver can then choose $\mathit{next}\;a$ as the existential witness.
Both generated facts carry applications of the original lemmas as Lean proofs.

\section{Formalization of Defunctionalization}
\label{app:defunctionalization-soundness}

We formalize typed higher-order and first-order languages and a translation between them, and prove Theorem~\ref{thm:defunctionalization-soundness}.

\subsection{Languages and Translation}
\label{app:formal-languages}

Fig.~\ref{fig:formal-languages} defines both languages and their semantics.
A higher-order signature assigns a type to each constant, and a context assigns
types to local variables. The Lean
representation is Intrinsically typed, \IE indexing terms by their signature, context, and type, so only
well-typed terms can be constructed. Formulas have type $\booltype$; they
are closed when their local context is empty.

First-order terms are likewise indexed by their context and result sort.
Generated symbols are identified by the original
constant, function type, or closure they represent, independently of concrete
SMT names. Their argument and result sorts are part of their types.
Valuations respect these type and sort assignments. In both languages,
quantifiers range over the full interpretation of their binder types or sorts,
and the valuation is extended with the bound value when evaluating the body.

\myparagraph{Recursive translation.}
The following rules give the full definition of $\defunc$, including the
cases summarized in Fig.~\ref{fig:defunctionalization-rules}.
Throughout, $\tau\in\{\booltype,\beta\}$ is a non-function sort and
$\theta=\sigma_1\to\cdots\to\sigma_n\to\tau$ with $n\geq1$.
Variables and constants of non-function sort, and Boolean literals, require
no recursive translation:
\[
  \inferrule{x:\tau}{\defunc(x)=(x,\emptyset,\emptyset)}
  \qquad
  \inferrule{c:\tau}{\defunc(c)=(\widehat c,\emptyset,\emptyset)}
\]
\[
  \inferrule{v\in\{\mathsf{true},\mathsf{false}\}}
       {\defunc(v)=(v,\emptyset,\emptyset)}
\]

\myparagraph{Applications.}
The helper $\headapp_h(\bar u)=(t,C^h,E^h)$ applies a head $h:\theta$
to a complete list of first-order arguments $\bar u$, returning a term
and the head's auxiliary axioms. Here $c$ is a constant, $f$ a bound
variable, and $\ell$ a lambda; $u_i$ has sort $\sigma_i^\flat$.
\begin{equation}
\label{eq:headapp-constant-variable}
  \inferrule{c:\theta}
       {\headapp_c(\bar u)=(\widehat c(\bar u),\emptyset,\emptyset)}
  \qquad
  \inferrule{f:\theta}
       {\headapp_f(\bar u)=(\appsym{\theta}(f,\bar u),\emptyset,\emptyset)}
\end{equation}
\begin{equation}
\label{eq:headapp-lambda}
  \inferrule{\ell:\theta\qquad\defunc(\ell)=(\ell^\flat,C_\ell,E_\ell)}
       {\headapp_\ell(\bar u)=(\appsym{\theta}(\ell^\flat,\bar u),C_\ell,E_\ell)}
\end{equation}
The application rule translates every argument before applying the head:
\begin{equation}
\label{eq:defunc-application}
  \inferrule{
      h:\theta\qquad
      \defunc(a_i)=(a_i^\flat,C_{a_i},E_{a_i})\quad(1\leq i\leq n)\\
      \headapp_h(\bar a^\flat)=(t,C^h,E^h)
  }{
    \defunc(h\,a_1\cdots a_n)=
    \bigl(t,\;C^h\cup\textstyle\bigcup_i C_{a_i},\;
      E^h\cup\textstyle\bigcup_i E_{a_i}\bigr)
  }
\end{equation}

\myparagraph{Equality and logical connectives.}
Equality at function sort additionally generates the extensionality axiom
$\extaxiom{\theta}$:
\[
  \inferrule{
    u,v:\tau\qquad
    \defunc(u)=(u^\flat,C_u,E_u)\qquad
    \defunc(v)=(v^\flat,C_v,E_v)
  }{
    \defunc(u=v)=(u^\flat=v^\flat,\;C_u\cup C_v,\;E_u\cup E_v)
  }
\]
\[
  \inferrule{
    u,v:\theta\qquad
    \defunc(u)=(u^\flat,C_u,E_u)\qquad
    \defunc(v)=(v^\flat,C_v,E_v)
  }{
    \defunc(u=v)=
    (u^\flat=v^\flat,\;C_u\cup C_v,\;E_u\cup E_v\cup\{\extaxiom{\theta}\})
  }
\]
\[
  \extaxiom{\theta}=(\forall f,g:\fnsort{\theta}.\;
    (\forall\bar x.\;\appsym{\theta}(f,\bar x)
      =\appsym{\theta}(g,\bar x))\to f=g).
\]
For logical connectives and quantifiers, let
\mbox{$\circ\in\{\land,\lor,\to,\leftrightarrow\}$} and $Q\in\{\forall,\exists\}$:
\[
  \inferrule{\defunc(p)=(p^\flat,C_p,E_p)}
       {\defunc(\neg p)=(\neg p^\flat,C_p,E_p)}
  \qquad
  \inferrule{\defunc(p)=(p^\flat,C_p,E_p)}
       {\defunc(Qx:\sigma.\,p)=(Qx:\sigma^\flat.\,p^\flat,C_p,E_p)}
\]
\[
  \inferrule{\defunc(p)=(p^\flat,C_p,E_p)\qquad
        \defunc(q)=(q^\flat,C_q,E_q)}
       {\defunc(p\circ q)=
          (p^\flat\circ q^\flat,C_p\cup C_q,E_p\cup E_q)}
\]
\myparagraph{Opening function bodies.}
To cover partial applications and lambdas whose bodies still have function
type, define a helper $\bodytrans(e,\bar z)=(r,C,E)$ that translates $e$
applied to fresh parameters $\bar z$ until its result is non-function.
The parameter types are the complete argument list of the type of $e$;
$\bar z$ is empty when $e$ has non-function type. Substitution below is
capture-avoiding.
\begin{equation}
\label{eq:bodytrans-ground-lambda}
  \inferrule{e:\tau\qquad\defunc(e)=(r,C,E)}
       {\bodytrans(e,())=(r,C,E)}
  \qquad
  \inferrule{\bodytrans(b[z_0/x],\bar z)=(r,C,E)}
       {\bodytrans(\lambda x.\,b,(z_0,\bar z))=(r,C,E)}
\end{equation}
For a non-lambda term $h\,\bar a$, the remaining parameters complete the
application. The list $\bar a=(a_1,\ldots,a_m)$ may be empty for a variable
or constant, and $\bar z$ contains the remaining $n-m$ parameters.
\begin{equation}
\label{eq:bodytrans-application}
  \inferrule{
      h:\theta\qquad m<n\\
      \defunc(a_i)=(a_i^\flat,C_{a_i},E_{a_i})\quad(1\leq i\leq m)\\
      \headapp_h(\bar a^\flat,\bar z)=(r,C^h,E^h)
  }{
    \bodytrans(h\,\bar a,\bar z)=
    \bigl(r,\;C^h\cup\textstyle\bigcup_i C_{a_i},\;
      E^h\cup\textstyle\bigcup_i E_{a_i}\bigr)
  }
\end{equation}
Fresh parameters are passed as first-order variables, including when
their types are function types.

\myparagraph{Closure construction.}
For any function-valued $e:\theta$, let $\ell$ be its lambda form obtained
by opening leading lambda binders and $\eta$-expanding over any remaining
arguments. Let $\bar y$ be its free variables and $\bar z$ fresh parameters
for all arguments of $\theta$.
\begin{equation}
\label{eq:defunc-closure}
  \inferrule{
    e:\theta\qquad \ell=\lambda \bar x.\ e\ \bar x \qquad\bar y=\mathsf{FV}(\ell)\qquad
    \bodytrans(e,\bar z)=(r,C,E)
  }{
    \defunc(e)=(\clossym{\ell}(\bar y),\;C\cup\{K_\ell\},\;E)
  }
\end{equation}
\[
  K_\ell=\forall\bar y,\bar z.\;
    \appsym{\theta}(\clossym{\ell}(\bar y),\bar z)=r.
\]
This generalizes the lambda clause of
Fig.~\ref{fig:defunctionalization-rules} and defines the function-valued
variable, constant, and partial-application cases. Captures are passed
directly as bound variables, without $\eta$-expansion. Opening nested lambda
binders together avoids intermediate closures.


\subsection{Model Extension}

For any higher-order model $M$, we construct its canonical first-order extension
$M^\star$ by retaining the carriers of base types and interpreting $\fnsort{\theta}$ by the
function space interpreting $\theta$. The interpretation of $\appsym{\theta}$
applies its function argument to the supplied arguments.
The interpretation of $\widehat c$ is the uncurried form of $M(c)$.
A closure symbol reconstructs the captured environment and evaluates the
corresponding lambda. This defines one interpretation for all generated
symbols, independent of the particular formula being translated (see Definition~\ref{def:canonical-extension}).

\subsection{Proof of Preservation of Denotations}
\label{app:proof-denotation}

\begin{proof}[of Lemma~\ref{lem:denotation-preservation}]
We fix a higher-order model $M$ and its extension $M^\star$ from
Definition~\ref{def:canonical-extension}. We first establish 3 valuation
identities by structural induction. For a type-preserving renaming $\pi$,
write $e[\pi]$ for the renamed term and
$(\rho\circ\pi)(x)=\rho(\pi(x))$. Then
\begin{equation}
\label{eq:denote-renaming}
  \denote{e[\pi]}{M}{\rho}=\denote{e}{M}{\rho\circ\pi}.
\end{equation}
Agreement on free variables gives the rule
\begin{equation}
\label{eq:denote-agreement}
  \inferrule{\forall x\in\mathsf{FV}(e).\;\rho_1(x)=\rho_2(x)}
       {\denote{e}{M}{\rho_1}=\denote{e}{M}{\rho_2}}.
\end{equation}
In the binder cases, the induction hypothesis applies after extending both
valuations with the same value. These identities also hold for FO terms.

For fresh $z$, the lambda and application semantics give
\[
  \begin{array}{rcl}
  \denote{\lambda z.\,e\,z}{M}{\rho}
    &=& (\mathsf{fun}\;v\Rightarrow\denote{e}{M}{\rho}\,v)\\
    &=& \denote{e}{M}{\rho}.
  \end{array}
\]
Repeated $\eta$-expansion therefore preserves the denotation of the lambda
form $\ell$ used in Eq.~(\ref{eq:defunc-closure}).

We prove preservation for $\defunc$ simultaneously with the following
properties of the two helpers, by induction on the translation derivations.
Each property quantifies over every valuation $\rho$ and $\bar w$:
\begin{equation}
\label{eq:headapp-preservation}
  \inferrule{\headapp_h(\bar u)=(t,C,E)}
       {\denote{t}{M^\star}{\rho}=
        \denote{h}{M}{\rho}\,
        \denote{u_1}{M^\star}{\rho}\cdots\denote{u_n}{M^\star}{\rho}}.
\end{equation}
\begin{equation}
\label{eq:bodytrans-preservation}
  \inferrule{\bodytrans(e,\bar z)=(r,C,E)}
       {\denote{r}{M^\star}{\rho[\bar z\mapsto\bar w]}=
        \denote{e}{M}{\rho}\,w_1\cdots w_n}.
\end{equation}
In Eq.~(\ref{eq:bodytrans-preservation}), $\bar z$ is fresh for $e$ and
$\bar w$ is any tuple of values of the corresponding sorts. For an empty
tuple, application to $\bar w$ leaves the value unchanged.

\myparagraph{Base Case: Variables, Constants, and Literals.}
For $x,c:\tau$ where $\tau$ is a non-function type, the base translation rules and
Fig.~\ref{fig:formal-languages} yield
\[
  \begin{array}{rcl}
  \denote{x^\flat}{M^\star}{\rho}&=\denote{x}{M}{\rho}=&\rho(x)\\
  \denote{c^\flat}{M^\star}{\rho}&=\denote{c}{M}{\rho}=& M(c)
  \end{array}
\]
The Boolean-literal rule preserves $\mathsf{True}$ and $\mathsf{False}$.
For $v_i=\denote{u_i}{M^\star}{\rho}$, their interpretations give
\[
  \begin{array}{rcl}
  \denote{\widehat c(\bar u)}{M^\star}{\rho}
    &=&M(c)\,v_1\cdots v_n,\\[2pt]
  \denote{\appsym{\theta}(f,\bar u)}{M^\star}{\rho}
    &=&\rho(f)\,v_1\cdots v_n.
  \end{array}
\]
This completes the proof of Eq.~(\ref{eq:headapp-preservation}) for constant and variable head (Eq.~(\ref{eq:headapp-constant-variable})).

\myparagraph{Closure Case.}
The induction hypothesis is Eq.~(\ref{eq:bodytrans-preservation}) for
the recursive call $\bodytrans(e,\bar z)$. Preservation of the closure term
itself follows directly from its interpretation.
In Eq.~(\ref{eq:defunc-closure}), let $\bar v=\rho(\bar y)$.
Since $\bar y=\mathsf{FV}(\ell)$, the captured valuation
$\rho_{\ell,\bar v}$ agrees with $\rho$ on every free variable of $\ell$.
Thus
\[
  \begin{array}{rcll}
  \denote{e^\flat}{M^\star}{\rho}
    &=&M^\star(\clossym{\ell})(\bar v)
    &\mbox{Eq.~(\ref{eq:defunc-closure}), Fig.~\ref{fig:formal-languages}}\\[2pt]
    &=&\denote{\ell}{M}{\rho_{\ell,\bar v}}
    &\mbox{Definition~\ref{def:canonical-extension}}\\[2pt]
    &=&\denote{\ell}{M}{\rho}
    &\mbox{Eq.~(\ref{eq:denote-agreement})}\\[2pt]
    &=&\denote{e}{M}{\rho}
    &\mbox{$\eta$-expansion}.
  \end{array}
\]
This covers lambdas and the $\eta$-expanded forms of function-valued
variables, constants, and partial applications.

\myparagraph{Lambda Case.}
The induction hypothesis for the recursive call $\defunc(\ell)$ is
$\denote{\ell^\flat}{M^\star}{\rho}=\denote{\ell}{M}{\rho}$ for every
$\rho$. Let $v_i=\denote{u_i}{M^\star}{\rho}$. The rule in
Eq.~(\ref{eq:headapp-lambda}) and the interpretation of $\appsym{\theta}$ give
\[
  \begin{array}{rcl}
  \denote{\appsym{\theta}(\ell^\flat,\bar u)}{M^\star}{\rho}
    &=&\denote{\ell^\flat}{M^\star}{\rho}\,v_1\cdots v_n\\[2pt]
    &=&\denote{\ell}{M}{\rho}\,v_1\cdots v_n.
  \end{array}
\]
This establishes Eq.~(\ref{eq:headapp-preservation}) for a lambda head.

\myparagraph{Application Case.}
The induction hypotheses are
$\denote{a_i^\flat}{M^\star}{\rho}=\denote{a_i}{M}{\rho}$ for every
argument $a_i$, and Eq.~(\ref{eq:headapp-preservation}) for
$\headapp_h(\bar a^\flat)$, all at every valuation $\rho$.
Applying them to Eq.~(\ref{eq:defunc-application}) gives
\[
  \begin{array}{rcl}
  \denote{(h\,\bar a)^\flat}{M^\star}{\rho}
    &=&\denote{h}{M}{\rho}\,
       \denote{a_1^\flat}{M^\star}{\rho}\cdots\denote{a_n^\flat}{M^\star}{\rho}\\[2pt]
    &=&\denote{h}{M}{\rho}\,
       \denote{a_1}{M}{\rho}\cdots\denote{a_n}{M}{\rho}\\[2pt]
    &=&\denote{h\,\bar a}{M}{\rho}.
  \end{array}
\]

\myparagraph{Equality Case.}
The induction hypotheses, for every $\rho$, are
\[
  \begin{array}{rcl}
  \denote{u^\flat}{M^\star}{\rho}&=&\denote{u}{M}{\rho},\\[2pt]
  \denote{v^\flat}{M^\star}{\rho}&=&\denote{v}{M}{\rho}.
  \end{array}
\]
The equality translation rule and Fig.~\ref{fig:formal-languages} give
\[
  \begin{array}{rcl}
  \denote{u^\flat=v^\flat}{M^\star}{\rho}
    &=&(\denote{u^\flat}{M^\star}{\rho}=\denote{v^\flat}{M^\star}{\rho})\\[2pt]
    &=&(\denote{u}{M}{\rho}=\denote{v}{M}{\rho})\\[2pt]
    &=&\denote{u=v}{M}{\rho}.
  \end{array}
\]
This applies to equality at both function and non-function sorts.

\myparagraph{Negation Case.}
The induction hypothesis is
$\denote{p^\flat}{M^\star}{\rho}=\denote{p}{M}{\rho}$ for every $\rho$.
By the negation translation rule and its semantics,
\[
  \begin{array}{rcl}
  \denote{\neg p^\flat}{M^\star}{\rho}
    &=&\neg\denote{p^\flat}{M^\star}{\rho}\\[2pt]
    &=&\neg\denote{p}{M}{\rho}\\[2pt]
    &=&\denote{\neg p}{M}{\rho}.
  \end{array}
\]

\myparagraph{Binary Connective Case.}
The induction hypotheses are preservation for $p$ and $q$ at every
$\rho$. For $\circ\in\{\land,\lor,\to,\leftrightarrow\}$,
\[
  \begin{array}{rcl}
  \denote{p^\flat\circ q^\flat}{M^\star}{\rho}
    &=&\denote{p^\flat}{M^\star}{\rho}\circ\denote{q^\flat}{M^\star}{\rho}\\[2pt]
    &=&\denote{p}{M}{\rho}\circ\denote{q}{M}{\rho}\\[2pt]
    &=&\denote{p\circ q}{M}{\rho}.
  \end{array}
\]

\myparagraph{Quantifier Case.}
The induction hypothesis is preservation for $p$ at every extended valuation:
\[
  \forall\rho,v.\quad
  \denote{p^\flat}{M^\star}{\rho[x\mapsto v]}
    =\denote{p}{M}{\rho[x\mapsto v]}.
\]
For $Q\in\{\forall,\exists\}$, corresponding sorts have the same carriers (\IE $\typedenote{\sigma^\flat}{M^\star}=\typedenote{\sigma}{M}$),
so the quantifier translation rule gives
\[
  \begin{array}{rcl}
  \denote{Qx:\sigma^\flat.\,p^\flat}{M^\star}{\rho}
    &=&Qv:\typedenote{\sigma^\flat}{M^\star}.\;
       \denote{p^\flat}{M^\star}{\rho[x\mapsto v]}\\[2pt]
    &=&Qv:\typedenote{\sigma}{M}.\;
       \denote{p}{M}{\rho[x\mapsto v]}\\[2pt]
    &=&\denote{Qx:\sigma.\,p}{M}{\rho}.
  \end{array}
\]

\myparagraph{Non-function Body Case.}
The induction hypothesis for $\defunc(e)=(r,C,E)$ is
\[
  \forall\rho.\quad\denote{r}{M^\star}{\rho}=\denote{e}{M}{\rho}.
\]
The non-function rule in Eq.~(\ref{eq:bodytrans-ground-lambda}) has no fresh
parameters, so this is exactly Eq.~(\ref{eq:bodytrans-preservation}).

\myparagraph{Lambda Body Case.}
The induction hypothesis for $\bodytrans(b[z_0/x],\bar z)=(r,C,E)$ is
\[
  \forall\rho',\bar w.\quad
  \denote{r}{M^\star}{\rho'[\bar z\mapsto\bar w]}
    =\denote{b[z_0/x]}{M}{\rho'}\,\bar w.
\]
Here, application to a tuple abbreviates successive application.
To prove Eq.~(\ref{eq:bodytrans-preservation}) for $\bodytrans(\lambda x.\,b,(z_0,\bar z))=(r,C,E)$, we need to prove
$$\forall\rho,w_0,\bar w.\quad\denote{r}{M^\star}{\rho[z_0\mapsto w_0,\bar z\mapsto\bar w]}=\denote{\lambda x.b}{M}{\rho}\,w_0\,\bar w.$$
Let $\rho'=\rho[z_0\mapsto w_0]$. The lambda rule in
Eq.~(\ref{eq:bodytrans-ground-lambda}), Eq.~(\ref{eq:denote-renaming}) give
\[
  \begin{array}{rcl}
  \denote{r}{M^\star}{\rho[z_0\mapsto w_0,\bar z\mapsto\bar w]}
    &=&\denote{r}{M^\star}{\rho'[\bar z\mapsto\bar w]}\\[2pt]
    &=&\denote{b[z_0/x]}{M}{\rho'}\,\bar w\\[2pt]
    &=&\denote{b[z_0/x]}{M}{\rho[z_0\mapsto w_0]}\,\bar w\\[2pt]
    &=&\denote{b}{M}{\rho[x\mapsto w_0]}\,\bar w\\[2pt]
    &=&\denote{\lambda x.\,b}{M}{\rho}\,w_0\,\bar w.
  \end{array}
\]
\myparagraph{Partial Application Body Case.}
The induction hypotheses are preservation for each supplied argument $a_i$
and Eq.~(\ref{eq:headapp-preservation}) for
$\headapp_h(\bar a^\flat,\bar z)$, at every valuation.
For Eq.~(\ref{eq:bodytrans-application}), let
$\rho'=\rho[\bar z\mapsto\bar w]$. These hypotheses and freshness of
$\bar z$ give
\[
  \begin{array}{rcl}
  \denote{r}{M^\star}{\rho'}
    &=&\denote{h}{M}{\rho'}\,
       \denote{a_1^\flat}{M^\star}{\rho'}\cdots
       \denote{a_m^\flat}{M^\star}{\rho'}\,\bar w\\[2pt]
    &=&\denote{h\,\bar a}{M}{\rho}\,\bar w.
  \end{array}
\]
This establishes Eq.~(\ref{eq:bodytrans-preservation}) and completes the
recursive cases.\qed
\end{proof}

\subsection{Proof of Validity of Generated Axioms}
\label{app:proof-axioms}

\begin{proof}[of Lemma~\ref{lem:generated-axioms}]
We fix $M$ and proceed by induction on the translation derivations: $M^\star\models C\cup E$ whenever
$\defunc$, $\headapp$, or $\bodytrans$ returns auxiliary sets $C,E$.
We use the following consequence of
Definition~\ref{def:denotational-semantics} to combine these sets:
\[
  \inferrule{M^\star\models S_1\quad\cdots\quad M^\star\models S_k}
       {M^\star\models\bigcup_{i=1}^k S_i}.
\]

\myparagraph{Base Case: Variables, Constants, and Literals.}
The base rules for $\defunc$, and the constant and variable head rules for
$\headapp$, generate $C=E=\varnothing$ (\IE $\top$). Thus
$M^\star\models C\cup E$ holds vacuously.

\myparagraph{Closure Case.}
The induction hypothesis for $\bodytrans(e,\bar z)=(r,C,E)$ is
$M^\star\models C\cup E$. It remains to establish the newly generated
closure equation $K_\ell$.
For Eq.~(\ref{eq:defunc-closure}), take any valuation $\rho$ and let
$\bar v=\rho(\bar y)$ and $\bar w=\rho(\bar z)$. Then
\[
  \begin{array}{rcl}
  \denote{\appsym{\theta}(\clossym{\ell}(\bar y),\bar z)}{M^\star}{\rho}
    &=&\denote{\ell}{M}{\rho_{\ell,\bar v}}\,\bar w\\[2pt]
    &=&\denote{e}{M}{\rho}\,\bar w\\[2pt]
    &=&\denote{r}{M^\star}{\rho}.
  \end{array}
\]
The first equality unfolds Definition~\ref{def:canonical-extension}; the
second uses Eq.~(\ref{eq:denote-agreement}) and $\eta$-expansion; the third
uses Eq.~(\ref{eq:bodytrans-preservation}). By equality and universal
quantification in Fig.~\ref{fig:formal-languages},
\[
  \inferrule{\forall\rho.\;
    \denote{\appsym{\theta}(\clossym{\ell}(\bar y),\bar z)}{M^\star}{\rho}
      =\denote{r}{M^\star}{\rho}}
       {M^\star\models\{K_\ell\}}.
\]
Combining this with the induction hypothesis gives
$M^\star\models C\cup\{K_\ell\}\cup E$, as required by the closure rule.

\myparagraph{Function Equality Case.}
The induction hypotheses are $M^\star\models C_u\cup E_u$ and
$M^\star\models C_v\cup E_v$. It remains to establish the newly generated
axiom $\extaxiom{\theta}$.
Let $F,G\in\typedenote{\fnsort{\theta}}{M^\star}$ satisfy the premise
of $\extaxiom{\theta}$. Definition~\ref{def:canonical-extension} interprets
these values as functions in $\typedenote{\theta}{M}$ and
$\appsym{\theta}$ as their application. Then,
\[
  \inferrule{\forall w_1,\ldots,w_n.\;
          F\,w_1\cdots w_n=G\,w_1\cdots w_n}
       {F=G}\quad\mbox{(function extensionality)}.
\]
Since $F$ and $G$ are arbitrary, the implication and quantifier semantics
give $M^\star\models\{\extaxiom{\theta}\}$. Together with the induction
hypotheses, this yields
\[
  M^\star\models C_u\cup C_v\cup E_u\cup E_v\cup\{\extaxiom{\theta}\}.
\]

\myparagraph{Application and Partial Application Body Cases.}
The induction hypotheses are $M^\star\models C^h\cup E^h$ for the head
and $M^\star\models C_{a_i}\cup E_{a_i}$ for every supplied argument.
Taking their union gives
\[
  M^\star\models
  \bigl(C^h\cup\textstyle\bigcup_i C_{a_i}\bigr)
  \cup\bigl(E^h\cup\textstyle\bigcup_i E_{a_i}\bigr),
\]
which is the output of Eqs.~(\ref{eq:defunc-application}) and
(\ref{eq:bodytrans-application}).

\myparagraph{Non-function Equality and Binary Connective Cases.}
The induction hypotheses give validity of the auxiliary sets for each
operand. These rules return their union, whose validity follows from the
union rule above.

\myparagraph{Negation and Quantifier Cases.}
The induction hypothesis is $M^\star\models C_p\cup E_p$ for the body
$p$. Both rules retain these same sets. The generated axioms are already
closed, so their validity is independent of the quantifier's bound value.

\myparagraph{Lambda and Non-function Body Cases.}
The induction hypothesis is validity of the sets returned by
$\defunc(\ell)$ or $\defunc(e)$, respectively. The corresponding rules
for $\headapp$ and $\bodytrans$ return those sets unchanged.

\myparagraph{Lambda Body Case.}
The induction hypothesis gives $M^\star\models C\cup E$ for the sets
returned by $\bodytrans(b[z/x],\bar z)$. The lambda rule in
Eq.~(\ref{eq:bodytrans-ground-lambda}) retains these sets; renaming
preserves validity by Eq.~(\ref{eq:denote-renaming}). This completes all
cases.\qed
\end{proof}

\subsection{Proof of Soundness}
\label{app:proof-soundness}

\begin{proof}[of Theorem~\ref{thm:defunctionalization-soundness}]
We suppose $\theoryencoding(T)$ is unsatisfiable and, for contradiction, assume a
higher-order model $M$ satisfying $T$. We construct $M^\star$ as in
Definition~\ref{def:canonical-extension}. By satisfaction and
Lemma~\ref{lem:denotation-preservation}, for each $\varphi\in T$ we have:
\[
  \inferrule{
    M\models T\qquad\varphi\in T\qquad
    \denote{\varphi^\flat}{M^\star}{\varnothing}
      =\denote{\varphi}{M}{\varnothing}}
       {M^\star\models\{\varphi^\flat\}}.
\]
Lemma~\ref{lem:generated-axioms} supplies the auxiliary formulas.
The same $M^\star$ interprets the symbols of every translated formula, so
Definition~\ref{def:complete-encoding} gives
\[
  \inferrule{
    \forall\varphi\in T.\;M^\star\models\{\varphi^\flat\}
    \qquad
    \forall\varphi\in T.\;M^\star\models C_\varphi\cup E_\varphi}
       {M^\star\models\theoryencoding(T)}.
\]
This contradicts the assumed unsatisfiability:
\[
  \inferrule{M^\star\models\theoryencoding(T)\qquad
        \neg\exists N.\;N\models\theoryencoding(T)}
       {\bot}.
\]
Hence no $M$ satisfies $T$, and $T$ is unsatisfiable.\qed
\end{proof}

\subsection{Scope of Formalization}
\label{app:formalization-boundaries}

The language is monomorphic and nondependent. It contains no datatype
constructors, pattern matching, or built-in arithmetic and bitvector operations.
Such types and operations can be treated abstractly as base types and constants, but their concrete SMT encodings are not justified by this proof. Each base type has its own carrier, so the formalization also excludes embeddings into a larger sort and the guarded closure, application-result, and extensionality axioms used for those embeddings. Solver-native higher-order translation is not modeled.

There are two gaps remaining in the formalization.
First, no mechanized result connects a Lean expression to a formula of this
higher-order language.
Second, first-order models interpret the entire
collection of structurally identified symbols, whereas an SMT script declares
a finite set of named symbols. Connecting these representations requires
showing that all symbols occurring in the generated formulas are declared and
that distinct structural symbols receive distinct names.
We aim to complete these formalizations in future work.

\section{Reconstruction Procedures}
\label{app:reconstruction-procedures}

\begin{figure}[t!]
\centering
\begin{minipage}{\linewidth}
\newcommand{\rcstep}[2]{\textcolor{gray}{\tiny #1}& #2\\}
\newcommand{\rcindent}{\hspace{0.7em}}
\hrule
\vspace{4pt}
\begin{minipage}[t]{0.49\linewidth}
\figurepart{Core-guided reconstruction}
\footnotesize
\renewcommand{\arraystretch}{1.2}
\begin{tabular}{@{}r@{\hspace{4pt}}l@{}}
\rcstep{1}{\textbf{procedure} $\mathsf{Core}(G,F,U)$}
\rcstep{2}{\rcindent $\Gamma\gets\{F(i):A_i\mid i\in U,$}
\rcstep{3}{\rcindent\rcindent\rcindent $F(i)\mbox{ is defined}\}$}
\rcstep{4}{\rcindent $p\gets\mathsf{Search}(\Gamma,G,b)$}
\rcstep{5}{\rcindent \textbf{require} $\Gamma\vdash_{\mathrm{kernel}}p:G$}
\rcstep{6}{\rcindent \textbf{return} $p$}
\multicolumn{2}{l}{}\\[-4pt]
\rcstep{7}{\textbf{procedure} $\mathsf{Search}(\Gamma,A,k)$}
\rcstep{8}{\rcindent \textbf{for} $t\in\mathcal T$ \textbf{do}}
\rcstep{9}{\rcindent\rcindent \textbf{try return} $t(\Gamma\vdash A)$}
\rcstep{10}{\rcindent \textbf{if} $k=0$ \textbf{then fail}}
\rcstep{11}{\rcindent \textbf{for} $r\in\mathcal R$ \textbf{of the form}}
\rcstep{12}{\rcindent\rcindent $r:\forall\bar x.\,A_1\to\cdots\to A_n\to B$}
\rcstep{13}{\rcindent \textbf{and} $\sigma$ \textbf{with} $\sigma(B)\equiv A$ \textbf{do}}
\rcstep{14}{\rcindent\rcindent \textbf{try}}
\rcstep{15}{\rcindent\rcindent\rcindent \textbf{for} $i=1,\ldots,n$ \textbf{do}}
\rcstep{16}{\rcindent\rcindent\rcindent\rcindent $p_i\gets\mathsf{Search}(\Gamma,\sigma(A_i),k-1)$}
\rcstep{17}{\rcindent\rcindent\rcindent \textbf{return} $r\,\sigma(\bar x)\,p_1\cdots p_n$}
\rcstep{18}{\rcindent \textbf{fail}}
\end{tabular}
\end{minipage}\hfill
\begin{minipage}[t]{0.49\linewidth}
\figurepart{Proof replay}
\footnotesize
\renewcommand{\arraystretch}{1.2}
\begin{tabular}{@{}r@{\hspace{4pt}}l@{}}
\rcstep{1}{\textbf{procedure} $\mathsf{Replay}(G,F,\Pi)$}
\rcstep{2}{\rcindent \textbf{assume} $h:\neg G$}
\rcstep{3}{\rcindent $F'\gets F[i_G\mapsto h]$; $P\gets\varnothing$}
\rcstep{4}{\rcindent \textbf{for input} $(i,C)$ \textbf{in} $\Pi$ \textbf{do}}
\rcstep{5}{\rcindent\rcindent $\Gamma_i\gets\{F'(i)\}$ \textbf{if defined,}}
\rcstep{6}{\rcindent\rcindent\rcindent $\varnothing$ \textbf{otherwise}}
\rcstep{7}{\rcindent\rcindent \textbf{try} $p\gets t(\Gamma_i\vdash\delta(C))$,}
\rcstep{8}{\rcindent\rcindent\rcindent $t\in\mathcal T$, \textbf{until success}}
\rcstep{9}{\rcindent\rcindent \textbf{require} $p:\delta(C)$; $P[i]\gets p$}
\rcstep{10}{\rcindent \textbf{for step} $(i,r,C,\bar j,\bar a)$ \textbf{in} $\Pi$ \textbf{do}}
\rcstep{11}{\rcindent\rcindent $A\gets\delta(C)$; $\bar p\gets\langle P[j]\mid j\in\bar j\rangle$}
\rcstep{12}{\rcindent\rcindent \textbf{for handler} $\eta\in\mathcal H_r$ \textbf{do}}
\rcstep{13}{\rcindent\rcindent\rcindent \textbf{try}}
\rcstep{14}{\rcindent\rcindent\rcindent\rcindent $\sigma\gets\mathsf{match}(\mathsf{pattern}_\eta,\bar a)$}
\rcstep{15}{\rcindent\rcindent\rcindent\rcindent $p\gets\eta(\sigma,\bar p;\ A)$}
\rcstep{16}{\rcindent\rcindent\rcindent\rcindent \textbf{require} $p:A$}
\rcstep{17}{\rcindent\rcindent\rcindent\rcindent $P[i]\gets p$; \textbf{break}}
\rcstep{18}{\rcindent\rcindent \textbf{if} $P[i]$ \textbf{undefined then fail}}
\rcstep{19}{\rcindent \textbf{require} $P[\square]:\mathsf{False}$}
\rcstep{20}{\rcindent $q\gets\mathsf{byContradiction}$}
\rcstep{21}{\rcindent\rcindent $(\mathsf{fun}\;h\Rightarrow P[\square])$}
\rcstep{22}{\rcindent \textbf{require} $q:G$ \textbf{by kernel checking}}
\rcstep{23}{\rcindent \textbf{return} $q$}
\end{tabular}
\end{minipage}
\vspace{4pt}
\hrule
\end{minipage}
\caption{Core-guided reconstruction (left) and proof replay (right).}
\label{fig:reconstruction-procedures}
\end{figure}

Fig.~\ref{fig:reconstruction-procedures} sketches reconstruction of a Lean goal
$G$. $F(i):A_i$ is the retained Lean proof for assertion $i$, $U$ is the UNSAT core, and $\Pi$ is the solver proof. $\mathcal R$ contains registered
reconstruction lemmas, $\mathcal T$ contains finishing tactics, and $b$
bounds backward search. A substitution $\sigma$ instantiates lemma parameters
or binds handler captures.
In the proof search procedure, \sys first tries the supplied finisher tactics (such as \texttt{omega}), given the premises in the UNSAT core.
If none of them can close the goal, \sys then tries to match the current goal
against the registered reconstruction lemma.
If a lemma's conclusion matches the goal, it tries to \texttt{apply} the lemma
and recursively search the proof for the premises of the lemma. 

In replay of solver inference steps, $\delta$ decodes SMT clauses using term handlers;
$\mathcal H_r$ includes handlers for rule $r$.
A step records its identifier ($i$), rule ($r$), clause ($C$),
premise identifiers ($\overline{j}$), and arguments ($\overline{a}$). $i_G$ is the negated goal.
We use $P[\square]$ to denote the proof of the final empty clause.
For each such step, \sys uses $\delta(C)$ to obtain the corresponding Lean obligation $A$.
Then, it tries to match the pattern against the arguments.
For a matching handler $\eta$, \sys uses $\eta$ to construct a proof term $p$ and, together with the previously constructed proof terms $\overline{p}$, checks if the reconstructed proof term has type $A$.
If any handler successfully closed the Lean obligation, \sys records the constructed proof term and moves to the next step. 
If all handlers failed, \sys will report the failure.
\section{Trusted SMT Comparison}
\label{app:trusted-comparison}

\begin{table}[t!]
\centering
\caption{Head-to-head comparison with \sys, trusting SMT verdicts.
Times are mean $\pm$ standard deviation across VCs discharged
by both tools, after averaging repetitions per VC.}
\label{tab:eval-head-to-head-trusted}
\footnotesize
\setlength{\tabcolsep}{3pt}
\renewcommand{\arraystretch}{1.12}
\begin{tabular}{@{}llrrrrr@{}}
\hline
& & \multicolumn{3}{c}{\textbf{\# VCs Discharged}}
& \multicolumn{2}{c}{\textbf{Time on Both (ms)}} \\
\cline{3-5}\cline{6-7}
\noalign{\vskip 3pt}
\textbf{Benchmark} & \textbf{Baseline}
& \shortstack{\sys\\\textbf{only}} & \shortstack{\textbf{Baseline}\\\textbf{only}}
& \textbf{Both} & \sys & \textbf{Baseline} \\
\hline
\multirow{4}{*}{\shortstack[l]{Curated\\(20 VCs)}}
& \leanauto & 2 & 0 & 17 & $234.8_{\pm 108.3}$ & $322.2_{\pm 20.1}$ \\
& Duper & 8 & 1 & 11 & $234.1_{\pm 114.0}$ & $98.3_{\pm 24.1}$ \\
& \leansmt & 7 & 0 & 12 & $239.4_{\pm 106.9}$ & $106.8_{\pm 31.5}$ \\
& \grind & 3 & 0 & 16 & $234.6_{\pm 112.8}$ & $5.5_{\pm 4.8}$ \\
\hline
\multirow{4}{*}{\shortstack[l]{Cashmere\\(38 VCs)}}
& \leanauto & 20 & 0 & 18 & $227.2_{\pm 381.8}$ & $294.0_{\pm 36.0}$ \\
& Duper & 17 & 0 & 21 & $247.0_{\pm 400.5}$ & $109.3_{\pm 102.6}$ \\
& \leansmt & 21 & 0 & 17 & $151.4_{\pm 341.1}$ & $135.6_{\pm 44.2}$ \\
& \grind & 19 & 0 & 19 & $215.4_{\pm 374.6}$ & $9.8_{\pm 2.7}$ \\
\hline
\multirow{4}{*}{\shortstack[l]{Velvet\\(504 VCs)}}
& \leanauto & 73 & 10 & 405 & $398.3_{\pm 490.2}$ & $378.6_{\pm 225.7}$ \\
& Duper & 190 & 3 & 288 & $232.0_{\pm 345.8}$ & $220.3_{\pm 517.4}$ \\
& \leansmt & 185 & 9 & 293 & $377.8_{\pm 480.7}$ & $145.5_{\pm 112.5}$ \\
& \grind & 43 & 9 & 435 & $349.5_{\pm 371.7}$ & $26.1_{\pm 26.8}$ \\
\hline
\multirow{4}{*}{\shortstack[l]{PLean\\(192 VCs)}}
& \leanauto & 2 & 0 & 172 & $3583.1_{\pm 11750.4}$ & $3621.0_{\pm 5972.5}$ \\
& Duper & 85 & 0 & 89 & $609.6_{\pm 268.2}$ & $2755.4_{\pm 8169.2}$ \\
& \leansmt & 74 & 0 & 100 & $1521.8_{\pm 8680.9}$ & $1298.7_{\pm 3712.7}$ \\
& \grind & 12 & 0 & 162 & $3528.6_{\pm 12102.2}$ & $1206.5_{\pm 2327.6}$ \\
\hline
\end{tabular}
\end{table}

Table~\ref{tab:eval-head-to-head-trusted} reports the head-to-head comparison
when \sys trusts SMT results. It uses the same VCs and baselines as
Table~\ref{tab:eval-head-to-head}. Runtime means and standard deviations are
computed separately for each pair over the VCs both tools discharge.


\end{document}